\documentclass[12pt]{article}
\usepackage[utf8]{inputenc}
\usepackage{fullpage}
\usepackage{amsmath}

\usepackage{amsthm}
\usepackage{mathtools}
\usepackage{bbm}
\usepackage{frcursive}
\usepackage{commath}
\usepackage{mathrsfs}
\usepackage{xcolor}
\usepackage{amssymb}
\usepackage{tikz}
\usepackage{enumitem}
\usepackage{hyperref}
\usepackage{comment}
\usepackage[UKenglish]{babel}
\usepackage{natbib}
\usepackage{array}
\usepackage{thm-restate}

\usepackage{algorithm}
\usepackage{algpseudocode}

\newcommand{\eps}{\varepsilon}

\newcommand{\E}{\mathbb{E}}
\renewcommand{\c}{\mathcal}
\renewcommand{\epsilon}{\varepsilon}
\newcommand{\vass}[1]{\left|#1\right|}

\renewcommand{\epsilon}{\varepsilon}

\newtheorem{theorem}{Theorem}[section]

\newtheorem{corollary}[theorem]{Corollary}

\newtheorem{proposition}[theorem]{Proposition}

\theoremstyle{definition}
\newtheorem{definition}[theorem]{Definition}

\newtheorem*{remark*}{Remark}

\title{The Bounds of Algorithmic Collusion\\
\vspace{.1in}
\large \centering $Q$-learning, Gradient Learning, and the Folk Theorem
\footnotetext{The authors are grateful to Panayotis Mertikopoulos, Christopher Sandmann, Ronny Razin, and Yannick Viossat for the useful discussions and insights they provided.}}
\author{Galit Askenazi-Golan\thanks{Department of Mathematics, London School of Economics. Email: \texttt{g.askenazi-golan, d.mergoni, e.plumb @lse.ac.uk}},\hspace{0.2cm} Domenico Mergoni Cecchelli\footnotemark[1],\hspace{0.2cm} Edward Plumb\footnotemark[1],\\ Clemens Possnig\thanks{Department of Economics, University of Waterloo. Email: \texttt{c.possnig (@) uwaterloo.ca}}}
\date{}
\begin{document}

\maketitle
\begin{abstract}
We explore the behaviour emerging from learning agents repeatedly interacting strategically for a wide range of learning dynamics, including $Q$-learning,  projected gradient, replicator and log-barrier dynamics. 
Going beyond the better understood classes of potential games and zero-sum games, we consider the setting of a general repeated game with finite recall under different forms of monitoring.
We obtain a Folk Theorem-style result and characterise the set of payoff vectors that can be obtained by these dynamics, discovering a wide range of possibilities for the emergence of algorithmic collusion. Achieving this requires a novel technical approach, which, to the best of our knowledge, yields the first convergence result for multi-agent $Q$-learning algorithms in repeated games.
\end{abstract}
%\tableofcontents

\section{Introduction}

Increasingly, strategic decisions are delegated to artificial intelligence (AI) agents. These agents operate in a possibly unknown environment and adapt their strategies over time while interacting with each other. As AI agents become more prevalent, it becomes necessary to understand the possible economic outcomes emerging from their interaction. One such outcome of concern is algorithmic collusion: the ability of learning agents to tacitly coordinate on high prices via reward–punishment schemes. Empirical evidence suggests that this outcome is possible across a variety of settings \cite{calvano2020artificial, assad2020algorithmic}.  

Our contribution is to provide a theoretical understanding of the set of possible outcomes emerging from repeated interactions among learning agents. We associate this set with the Folk Theorem in repeated games, showing that a rich set of outcomes can arise under a broad class of learning schemes. In particular, we focus on reinforcement learning (RL), which refers to updating rules that guide an agent's policy\footnote{We use `policy' and `strategy' interchangeably, as well as `agent' and `player'. In both cases, game theorists more commonly use the latter, while computer scientists lean toward the former.} towards actions believed to yield higher payoffs. The RL methods we study are $Q$-learning and the $q$-replicator family, where the latter encompasses gradient learning, replicator, and log-barrier dynamics as special cases.

$Q$-learning is a classical method designed to identify optimal policies in Markov decision processes (MDPs). It iteratively updates a $Q$-table, which estimates the expected future payoff after taking action $a$ in state $s$. Convergence of the $Q$-table corresponds to identifying the optimal policy in the MDP.

The $q$-replicator dynamics generalise gradient learning. Here, a policy mapping states to action probabilities is updated over time using the estimated rewards. The policy evolves in proportion to the relative profitability of actions, converging when no further improvement is possible, i.e., a local maximum is reached.

Both methods are successful and well-understood in MDPs, which explains their popularity in multi-agent settings. However, multi-agent learning introduces additional challenges: the environment is non-stationary due to the simultaneous adaptation of other agents' strategies, and the value of any strategy depends on the strategies of all other agents. The dynamics in the policy space induced by this interaction may be chaotic, cyclical, and nonconvergent.

Our framework accommodates both perfect and imperfect monitoring. The characterisation of learnable strategy profiles\footnote{That is, profiles that can emerge as the limiting behaviour of the RL process.} provides new insights into solution concepts under different monitoring assumptions. Specifically, perfect public equilibria are attainable via RL, whereas sequential equilibria are not directly comparable to the set of learnable strategy profiles.

In repeated games, the analogue of a state in an MDP is the history of play. Therefore, applying RL to repeated games naturally entails representing finite histories as state variables. As machines are limited in their representation capabilities, only finite-valued state variables can be used in RL. Therefore, in this paper, we consider finite, recent histories as state variables both for $Q$-learning and $q$-replicator approaches, in other words, our agents learn to play finite-recall strategies.\footnote{Finite memory strategies are an alternative approach that is beyond the scope of this paper.} 

The RL agents are not assumed to have complete knowledge of game payoffs. Our procedures are estimation-based, so all results are probabilistic in nature. Under perfect monitoring, we show that strict subgame-perfect Nash equilibria with finite recall can be learnt with arbitrarily high probability. We prove that the set of learnable outcomes corresponds closely to the Folk Theorem for finite recall \cite{barlo2016bounded}. Specifically, all outcomes identified in \cite{barlo2016bounded} that can also be achieved by strict subgame-perfect Nash equilibria are learnable by the RL agents.\footnote{The main difference is due to Theorem 2 of \cite{barlo2016bounded}, which requires mixed minmax strategies. Such strategies cannot be strict subgame-perfect. We provide a version using pure minmax strategies.}

Beyond the updating rules, $Q$-learning and $q$-replicator differ in their temporal frameworks. $Q$-learning is typically a `continuing' learning process: agents interact repeatedly while updating their $Q$-tables ad infinitum. In contrast, $q$-replicator is usually `episodic': one realisation of the repeated interaction is an episode, and episodes are repeated infinitely.  Practically, this is achieved by taking the discount factor to reflect the probability of terminating an episode. This distinction generates a difference in what is an appropriate convergence criterion. For $q$-replicator, convergence to an equilibrium trivially implies convergence of accumulated rewards to the equilibrium payoff. For $Q$-learning however, convergence of the $Q$-table to one associated with an equilibrium does not necessarily guarantee convergence of realised payoffs during learning to the equilibrium payoff.

\subsection{Relation to the Literature}

To the best of our knowledge, this work presents the first Folk Theorem for learning in general finite-player, finite-action games, thereby extending the literature beyond potential games and zero-sum games (\cite{daskalakis2020independent}, \cite{mokhtari2020unified}, \cite{perolat2021poincare}, \cite{fox2022independent}, \cite{mguni2021learning} and \cite{cartea2025algorithmic}). This broader framework allows for a general perspective on what agents might learn, encompassing collusion, competition or other behaviours. 

Other theoretical works in the literature on algorithmic collusion focus on the emergence of cooperation among memoryless algorithms (\cite{banchio2022artificial}), the interplay of monitoring and stage games under continuous action settings (\cite{possnig2023reinforcement}), or consider more stylised models of algorithmic competition (\cite{lamba2022pricing}, \cite{brown2021competition}, \cite{johnson2020platform}, and \cite{salcedo2015pricing}). 

There are two main technical contributions. First, we establish the first general convergence result for multi-agent $Q$-learning processes converging to strict Nash equilibria (Theorem \ref{thm: perfmon_Qlearning}). Previous theoretical results either restricted attention to static interactions, where agents possess no memory (\cite{leslie2005individual}), or relied on substantially simplifying adjustments to the learning rules, such as coordinated exploration phases across agents (\cite{arslan2016decentralized}). 

Second, we develop a novel approach to studying convergence to Nash equilibria in repeated games. This analysis must account for the difficulties posed by pure equilibria, where certain histories may never be realised, causing standard ergodicity conditions to fail. Addressing this issue requires a careful generalisation of Theorem~2 in \cite{giannou2021convergence}.

Finally, our results extend existing convergence guarantees in stochastic games from gradient learning (\cite{giannou2021convergence}) to the broader class of $q$-replicator dynamics, as outlined above.

\section{The Game Model}\label{sec2_game_model}
%The aim of this section is to introduce our model of repeated games with imperfect monitoring  (of which perfect monitoring is a special case) and to specify our solution concepts. We start by defining the concept of stage game, we then introduce repeated games with imperfect monitoring and we conclude by introducing our solution concepts.

We begin by introducing the model. A stage game (or one-shot game) is an ordered triple $G = (N,A,(R_i)_{i \in N})$ where $N$ is the finite set of players, $A=\prod_{i\in N}A_i$ where $A_i$ is the finite set of actions available to player~$i \in N$, and $R_i:A \to \mathbb{R}$ is the reward function of player~$i \in N$.

For a given stage game $G$, a round of $G$ consists of an action profile $a=(a_1,\dots{},a_{\vass{N}})\in A$ and its corresponding reward vector $r=(r_1,\dots{r_{\vass{N}}})=(R_1(a), \dots{}, R_{\vass{N}}(a))$.

For any finite set $B$, let $\Delta(B)$ denote the set of probability distributions over $B$. For a fixed game $G$, we define for each player $i\in N$ their utility function $u_i : \prod_{j \in N} \Delta(A_j) \rightarrow \mathbb{R} $ by taking expected rewards under mixed action profiles.  Specifically, given a strategy profile $\pi \in \prod_{j \in N} \Delta(A_j)$, where each player $j$ independently draws an action according to $\pi_j$, player $i$'s utility is $u_i(\pi) = \mathbb{E}_{a \sim \pi}[R_i(a)]$. 

\subsection{The Repeated Game}

% We consider a model in which players play a stage game $G$ repeatedly over an indefinite number of periods. In our setting, there is a small fixed probability that the repeated game terminates at the end of each period. As known from the literature, this is equivalent to a repeated game with no probability of termination after each period, but with future rewards being exponentially discounted. 
%%%%%%%%%%%%%%%%%%%%

%We consider an infinitely repeated game with discounting. The model we propose allows for for full  The observation of each player is modelled as a random variable, called a private signal, whose distribution depends upon the action profile played. We denote by $q$ the function which maps the action profile played to the joint distribution over signal realisations observed by the players. In particular, the signals of the players may be correlated or independent. Public monitoring, in which all players observe the same signal, and perfect monitoring, in which each player's signal reveals the whole action profile played, are specific cases of our model. Players condition their actions upon their private history, which we assume is comprised of their own actions and signal realisations. Furthermore, we assume that, in every period, each player can deduce their reward from their individual action and signal realisations. We also assume that each player has a finite recall, which is a natural assumption in the context of reinforcement learning. 

We study an infinitely repeated game with discounting and imperfect monitoring. In each stage, after an action profile is played, players observe private signals whose joint distribution depends on the realised action profile and may exhibit arbitrary correlation. The associated signal spaces and conditional signal distributions define the game’s monitoring structure, that is, the information structure governing what players observe about past play. Our formulation encompasses public monitoring and perfect monitoring as special cases. At each stage, players choose actions based on their private histories, consisting of their own past actions and observed signals, from which they can infer their realised payoffs. Finally, we impose a finite-recall assumption on players' behaviour, motivated by reinforcement-learning considerations.

Formally, fix a stage game $G$. Let $Z=\prod_{i\in N}Z_i$ where $Z_i$ is a finite set of possible private signal realisations for player $i \in N$. Let $q:A\to \Delta(Z)$ be the joint distribution over signal realisations conditional on the action profile played, and let $\delta\in (0,1)$ be the discount factor.

The repeated game $\Gamma(\delta)$  with stage-game $G$ is an infinite-horizon stochastic process defined as follows.
At the beginning of period $t+1$, the current history is given by $h^t=(a^1, z^1, ...,a^t,z^t)$, where $a^k \in A$ and $z^k \in Z$ for each $k \leq t$. For $t=0$, the history is empty and is denoted by $h^0 = \varnothing$.  The set of all possible histories of the game up to and including period $t$ is $H^t=(A\times Z )^t$, with $H^0 = \{\varnothing\}$. For each player $i \in N$, the $i^{th}$ component of $h^t$ is denoted by $h^t_i$ and is called the private history of player $i$; similarly, $H^t_i = (A_i \times Z_i )^t$ is the set of all possible private histories of player $i$ up to and including period $t$, with $H_i^0 = \{\varnothing\}$\footnote{In the case of perfect monitoring, one may equivalently take $H^t = A^t$, which implies $H_i^t = A^t$ for all $i \in N$.}.

During period $t+1$, an action profile, a signal profile, and the corresponding rewards are realised.  Let $a^{t+1}=(a_1^{t+1},a_2^{t+1},...,{a_{|N|}}^{t+1}) \in A$ denote the action profile at period $t+1$, where we specify below how the action $a_i^{t+1}$ is chosen by player $i$. Conditional on $a^{t+1}$, a signal profile $z^{t+1}=(z_1^{t+1},z_2^{t+1},...,{z_{|N|}}^{t+1})$ is drawn according to $q(a^{t+1})$. Finally, each player $i$ receives a realised reward $r_i^{t+1} = R_i(a^{t+1})$ and we write $r^{t+1} = (r_1^{t+1}, \dots, r_{|N|}^{t+1})$ for the resulting reward vector.

%At the end of period $t+1$, a Bernoulli random variable with success probability $\delta$ is sampled independently. If not successful, we terminate the game and return $h:=h^{t+1}$ as a sample of the repeated game. Let $\tau(h)$ be the termination period. If successful, we repeat our last step. 

We use the terms strategy (as in the game-theoretic literature) and policy (as in the reinforcement learning literature) interchangeably. In standard models of repeated games, players typically condition their strategies on histories of unbounded length. Such an assumption is at odds with computability and memory constraints that are a fundamental part of a reinforcement learning analysis. We therefore restrict attention to finite-recall strategies for each player.

%For $l$ a positive integer, define $\hat{h}_i^l=(a_i^{t-l+1}, z_i^{t-l+1},..., a_i^t, z_i^t)$ as the private history of player $i$ of the last $l$ periods, called the $l$-recall history (the index $t$ is omitted for readability).  If $t < l$, then $\hat{h}_i^l=(a_i^{1}, z_i^{1},..., a_i^t, z_i^t)$. Then, $\hat H_i^l =(A_i \times Z_i )^l $ is the set of $l$-recall histories of player~$i$, and  $\hat{H}^\ell=\prod_{i\in N}{\hat{H}_i}^{\ell_i}$  for some $\ell \in \mathbb{Z}_{+}^{|N|}$ is  the set of $\ell$-recall histories. Finally, let $H^\infty=\bigcup_{t\in \mathbb{N}}(A\times Z)^t$ be the set of all possible (finite) histories, and $h^\infty=(a^1, z^1,a^2,z^2,...)\in H^\infty$ be a history, a play of the game.

 %Player~$i$ uses a strategy that conditions actions on $\hat{h}_i^{\ell_i}$: the private history of the last $\ell_i$ periods. The set of mixed strategies of player $i$, denoted $\Pi^{\ell_i}_i$, is $\Delta \left (A_i^{\hat{H}^{\ell_i}_i} \right )$. For any $\pi_i \in \Pi_i^{\ell_i}$, define $\rho_i(\pi_i):\hat H_i^{\ell_i}\to \Delta(A_i)$ as the behavioural strategy induced by $\pi_i$. Let $\rho(\pi)$ be the profile of such induced behavioural strategies. 

Fix a positive integer $\ell$. For each $t \ge 0$, define the $\ell$-recall private history of player $i$ at the beginning of period $t+1$ by
\begin{align*}
\hat h_i^{\ell} \;=\;
\begin{cases}
\bigl(a_i^{t-\ell+1}, z_i^{t-\ell+1}, \dots, a_i^{t}, z_i^{t}\bigr), & \text{if } t \ge \ell,\\[0.5em]
\bigl(a_i^{1}, z_i^{1}, \dots, a_i^{t}, z_i^{t}\bigr), & \text{if } 0 < t < \ell,\\[0.5em]
\varnothing, & \text{if } t = 0,
\end{cases}
\end{align*}
where the dependence on $t$ is omitted for notational simplicity. The set of all possible $\ell$-recall private histories of player $i$ is $\hat H_i^\ell = \bigcup_{k=0}^{\ell} (A_i \times Z_i)^k$, and the set of all $\ell$-recall histories\footnote{Our convergence results, in fact, allow for different players to have different lengths of recall.} is $\hat H^\ell = \prod_{i \in N} \hat H_i^\ell$. Finally, let $H^\infty=\bigcup_{t\in \mathbb{N}}(A\times Z)^t$ be the set of all possible (finite) histories, and $h^\infty=(a^1, z^1,a^2,z^2,...)\in H^\infty$ be a history, a play of the game.

Each player~$i$ uses a strategy that conditions actions on $\hat{h}_i^{\ell}$, the private history of the most recent $\ell$ periods. The set of mixed $\ell$-recall strategies of player $i$ is denoted by $\Pi_i^{\ell} = \Delta \left (A_i^{\hat{H}^{\ell}} \right )$. For any $\pi_i \in \Pi_i^{\ell}$, let $\rho_i(\pi_i):\hat H_i^{\ell}\to \Delta(A_i)$ denote the induced behavioural strategy\footnote{A strategy that assigns a distribution over the action set available at each state.}. We write $\rho(\pi)$ for the corresponding profile of behavioural strategies.

A strategy profile $\pi \in \Pi^{\ell}  \coloneqq \prod_{i \in N} \Pi_i^{\ell}$ induces a probability distribution over the set $H^{\infty}$ of realised histories. The expected discounted utility of player $i \in N$ under $\pi$ is given by

$$
V_i(\pi):=\E_{h^{\infty}\sim \pi}\left[\sum_{t=1}^{\infty}\delta^t R_i(a^{t})\right].
$$

\section{Solution Concepts}\label{sec sol}

We now introduce the solution concepts for the models presented in the previous section.

Our model allows for a wide range of monitoring structures. In the game theory literature, specific solution concepts are defined for certain settings of monitoring, recall, and admissible strategy types. A detailed discussion of how the solution concepts introduced here relate to standard equilibrium notions, such as sequential equilibrium and perfect public equilibrium, can be found in Section~\ref{sec: Impfmon}.

\subsection{Equilibrium in repeated games}

Given a repeated game with stage game $G$ and $\ell\in \mathbb{N}$, a strategy profile $\pi^*\in \Pi^{\ell} $ is an $\ell$-recall equilibrium if no player $i$ has a profitable unilateral deviation to any $\ell$-recall strategy.

\begin{definition}[$\ell$-recall equilibrium]
A strategy profile $\pi^*\in \Pi^{\ell}$ is an $\ell$-recall equilibrium if for any player $i\in N$ and any $\ell$-recall strategy  $\pi_i\in\Pi^{\ell}_i$, we have $V_{i}(\pi^*) \geq V_{i}(\pi_i,\pi_{-i}^*)$.
\end{definition}

We next introduce the subgame-perfect refinement of $\ell$-recall equilibrium. For any profile $\pi \in \Pi^{\ell}$ and $\hat h^{\ell}\in \hat H^{\ell}$, define $\rho(\pi \mid \hat h^{\ell})$ as the behavioural strategy profile of the subgame starting at history $\hat h^{\ell}$. Also define $h^{\infty}(\hat h^{\ell})$ as any infinite history that can follow $\hat h^{\ell}$. Then the expected utility of player $i$ conditional on having observed history $\hat h^{\ell}$ is
$$
V_i(\hat h^{\ell}, \pi) = \E_{h^{\infty}(\hat h^{\ell})\sim \rho(\pi \mid \hat h^{\ell})}\left[\sum_{t=0}^{\infty}\delta^t R_i(a^{(t)}) \right]
$$
\begin{definition}[$\ell$-recall subgame-perfect equilibrium]
A strategy profile $\pi^*\in \Pi^{\ell}$ is an $\ell$-recall  subgame-perfect equilibrium if, for every player $i \in N$, every $\ell$-recall history $\hat h_i^{\ell} \in \hat H_i^{\ell}$, and every $\ell$-recall strategy $\pi_i \in \Pi_i^{\ell}$, we have $V_{i}(\hat h^{\ell},\pi^*) \geq V_{i}(\hat h^{\ell},\pi_i,\pi_{-i}^*)$.
\end{definition} 
An $\ell$-recall subgame-perfect equilibrium is \emph{strict} if any unilateral deviation induces a strict loss.

The study of repeated games often focuses on the correspondence between Nash equilibria and the payoff vectors they generate. One aspect of this relationship is given by the celebrated Folk Theorem. Two important assumptions of the Folk Theorem are perfect monitoring,  and that the players have unbounded recall. Under these assumptions, the Folk Theorem characterises the set of payoff vectors that correspond to Nash equilibria as the set of feasible and individually rational payoffs, to be defined shortly. A complete discussion of how the results here extend to the imperfect monitoring case can be found in Section \ref{sec: Impfmon}.

 \cite{barlo2016bounded} establish that the set of payoff vectors under perfect monitoring and unbounded recall can be approximated by equilibrium payoff vectors when players are restricted to have finite recall. We invoke their Theorem 1, which considers any finite number of players and restricts the players to using pure minmax strategies.\footnote{As will become clear later on, Q-learners generically cannot learn to play mixed minmax strategies.}

Player $i$'s minmax value in pure strategies is
$$\tilde{u}_i := \min_{a_{-i} \in A_{-i}} \max_{a_i \in A_i} u_i(a_i, a_{-i}).$$

We denote the set of feasible and strictly individually rational payoffs as $\tilde{W}:=\{u\in\text{conv}\{u(a):a\in A\}: u_i> \tilde{u}_i\}$, where $\text{conv}\{u(a):a\in A\}$ is the convex hull of the set $\{u(a):a\in A\}$. 

Fix a stage game $G= (N, A, (R_i)_{i \in N})$, with its associated set of feasible and strictly individually rational payoffs $\tilde W$, and let $\Gamma(\delta)$ be the associated $\delta$-discounted repeated game with stage game $G$ and perfect monitoring.

We now state a variation of Theorem 1 in \cite{barlo2016bounded}. All proofs can be found in Appendix \ref{sec: proofs}. 

\begin{restatable}{theorem}{thmpayoffapprox}\label{thm: payoff_approx}
 For all $\epsilon>0$ there is $\delta^*\in (0,1)$ such that for all $\delta\in (\delta^*, 1)$ and every $u\in \tilde{W}$, there exists $\ell \in \mathbb{N}$ and an $\ell$-recall strict subgame-perfect equilibrium $\pi^*$ of $\Gamma(\delta)$ such that the distance between $u$ and  $V(\pi^*)$ is at most $\epsilon$.
\end{restatable}

Motivated by this result, we make the following definition:

\begin{definition}[$\varepsilon$-finite implementation]
Let $u \in \tilde W$ and $\varepsilon > 0$. We say that an $\ell$-recall strict subgame-perfect equilibrium $\pi^*$ of $\Gamma(\delta)$ is an \emph{$\varepsilon$-finite implementation} of $u$ if $\delta$ is sufficiently close to 1 and $\ell \in \mathbb{N}$ are such that the distance between $u$ and  $V(\pi^*)$ is at most $\epsilon$.
\end{definition}

Next, we divide our results into the two main classes of learning processes we consider: $Q$-learning, and $q$-replicator dynamics. 

\section{$Q$-learning: Dynamics and Results}

$Q$-learning is a reinforcement learning method that estimates the value of taking a given action in a given state and thereafter behaving greedily. These estimates are stored in a table, known as the $Q$-table. Each entry $Q(s,a)$ represents the agent’s current estimate of the expected discounted payoff obtained by choosing action $a$ in state $s$ and subsequently playing optimally.

For the sake of exposition, fix a finite state space $S$ and a finite action set $A$. During the learning process, the $Q$-table is updated only at state-action pairs that are actually realised. Specifically, when the agent is in state $s$, chooses action $a$, receives reward $r_t$, and transitions to a new state $s'$, the update rule is
\begin{align*}
Q_{t+1}(s,a)
&= Q_t(s,a)
 + \gamma_t\Bigl(
 r_t
 + \delta \max_{a' \in A} Q_t(s',a')
 - Q_t(s,a)
 \Bigr),
\end{align*}
where $\delta \in (0,1)$ is a discount factor and $\gamma_t>0$ is a stepsize sequence that decays over time.

In a standard version of $Q$-learning, players select actions in an $\eps$-greedy manner, where $\eps \in(0,1)$ is an exploration parameter: every period, with probability $\eps$, an action is sampled uniformly from $A$; otherwise, the action that maximises the current estimate $Q_t$ at the current state is played. In standard single player MDP settings, where the state variable satisfies an ergodicity property, it has been shown (\cite{watkins1992q}) that $Q_t$ will converge, in probability, to the solution of a Bellman equation. This solution characterises the optimal strategy in the MDP. 

Due to the success and simplicity of this method, it has also been commonly studied in multi-agent settings. Here, each agent $i$ follows their own table $Q_i$ independently. Clearly, the standard MDP assumptions, such as stationarity, fail. Despite this lack of stationarity, it is still possible to establish local convergence results. In particular, we believe ourselves to be the first to show that if the joint $Q$-learning dynamics enter a sufficiently small neighbourhood of a strict Nash equilibrium, then the learning process converges with high probability to $Q$-values that support the equilibrium strategy. 

\subsection{Main Results: $Q$-learning}

Fix a recall length $\ell \in \mathbb{N}$. When $\pi^*$ is an $\ell$-recall strict subgame-perfect Nash equilibrium of $\Gamma(\delta)$, we let $Q^*$ denote the corresponding $Q$-table of the players. We fix a stepsize sequence $(\gamma_t)_{t \geq 1}$ satisfying the usual stochastic approximation assumptions; in particular, we take $\gamma_t = \frac{\gamma}{t}$ for some $\gamma \in (0,1)$ where $t$ counts the number of visits to a given state-action pair\footnote{See Appendix \ref{Appendix_proofs_thm_Qlearning} for details}.

Given a $Q$-table $Q$, we write $\pi(Q)$ for the greedy policy induced by $Q$: after any history of $\Gamma(\delta)$, $\pi(Q)$ selects an action profile that maximises the corresponding $Q$-values.

\begin{restatable}{theorem}{thmPerfMonQLearning}\label{thm: perfmon_Qlearning}
Take $\eps>0$. For every $u\in \tilde{W}$, pick an $\frac{\eps}{2}$-finite implementation $\pi^*$. Denote by $Q^*$ the $Q$-table of $\pi^*$.
\begin{enumerate}
    \item For every $\eta>0$ there is a neighbourhood $\c U$ of $Q^*$ such that for any $Q_0\in \c U$ and any small enough $\gamma>0$ we have, 
$$
\mathbb{P}\left(\forall t>0,\  \pi(Q_t)=\pi(Q^*)\right )\geq 1-\eta.
$$
\item With probability at least $1-\eta$, the vector of expected discounted payoffs accrued by the $Q$-learning dynamics is $\epsilon$-close to $u$.
\end{enumerate}
\end{restatable}

Hence, Theorem~\ref{thm: perfmon_Qlearning} establishes that independent $Q$-learning dynamics can, with high probability, converge to strategy profiles that implement any payoff in $\tilde{W}$ up to an arbitrary $\epsilon$. In other words, the set of outcomes achievable through these dynamics closely approximates the full set of payoffs supported by the Folk Theorem.

\begin{comment}
\begin{theorem}
    For any $\varepsilon, \varepsilon' > 0$, there exists  $\eta, \gamma>0$ such that, if $Q^0$ is the initial $Q$-matrix for the players such that $\lVert Q^0 - Q^* \rVert_\infty < \eta$ and the exploration rate is $\frac{\gamma}{n}$, then with probability at least $1 - \varepsilon$ we have:
\begin{align*}
    \pi(Q^n) = \pi(Q^*) \text{ for all } n \geq 0,
\end{align*}
and the vector of discounted payoffs from the beginning of the learning is $\epsilon'$-close to the vector of discounted payoffs obtained playing $\pi^*$. 

\textbf{I think that (with the same probability), for all state-action pairs $(s,a)$ for some player $i$ that are visited infinitely often, we also have $Q_i^n(s,a) \rightarrow Q_i^*(s,a)$. }
\end{theorem}
\end{comment}
%discuss here or : one can adapt Barlo to a construction with finite path back to on-path. this will be needed when introducing the payoff convergence result

\section{$q$-Replicator: Dynamics and Results}

The $q$-replicator dynamic generalises the widely used gradient ascent procedure by introducing what we call the `$q$-gradient'. Here, the strategy of player $i$ is modified in the direction of the $q$-gradient of the expected reward of player $i$. The case where $q=0$ corresponds to the standard gradient and gives rise to a gradient ascent procedure. This generates a process in which every player modifies their strategy in a direction that improves their expected payoff, given fixed opponents' strategies.

\begin{definition}[$q$-gradient]\label{def:q-gradient}
Let $q\in \mathbb{R}_{\geq 0}$, and fix a player $i$ and a mixed strategy profile $\pi$. The $q$-gradient for player $i$ at $\pi$ is the vector $v_{q,i}$ with components $v_{q,i,\alpha}$ such that
$$
v_{q,i,\alpha}(\pi_i, \pi_{-i}) = (\pi_{i,\alpha})^q \left ( V_{i}(e_{\alpha}, \pi_{-i}) - \frac{ \sum_{\beta}(\pi_{i,\beta})^{q} V_{i}(e_{\beta}, \pi_{-i})}{\sum_{\beta}(\pi_{i,\beta})^q} \right ),
$$
where $e_\alpha$ is the pure strategy associated with the $\alpha$-th component of $\pi_i$.
\end{definition}

Note that the term in the parentheses is the surplus utility that action $\alpha$ at state $s$ obtains over the weighted average of the other actions. Furthermore, note that for the case $q=0$ we retrieve the standard definition of gradient minus a standard normalisation term.

 As we are interested in the learning setting where expected payoffs are unknown\footnote{This includes the case where payoffs are unknown and the case where payoffs are known, but not the other agent's strategies.}, we consider the situation where each agent $i$ at period $t$ observes an estimate $\hat v_{q,i}(\pi_t)$ of $v_{q,i}(\pi_t)$. Fix a stepsize sequence $(\gamma_t)_{t \geq 1}$. Then, the $q$-replicator dynamics unfold as follows:
 \begin{align}\label{eq: q-replicator}
     \pi_{i,t+1}=\text{proj}_{\Pi_{i}^{\ell}}(\pi_{i,t} + \gamma_{t} \hat{v}_{q,i}(\pi_{i,t}))
 \end{align}
where $\text{proj}_{\Pi_{i}^{\ell}}$ is the Euclidean projection to agent $i$'s policy space.

When $q=1$ and $q=2$, replacing the estimator with $v_{q,i}(\pi)$, these dynamics reduce to discrete-time variants of replicator and log-barrier dynamics, respectively (\cite{riemannain}). Additionally, in continuous time, these dynamics meet the criteria for myopic adjustment dynamics for \( q \geq 1 \) as defined by \cite{swinkels1993adjustment}, and satisfy positive correlation for \( q \geq 0 \) (\cite{sandholm2010population, riemannain}).

\subsection{Main Results: $q$-replicator}

The behaviour of the $q$-replicator process depends only on the initialisation $\pi_0$, the stepsize schedule $(\gamma_t)_{t\geq 1}$, and the properties of the estimators $\hat v_{q,i,t}$. The sufficient conditions we impose on the stepsizes and on the properties of $\hat v_{q,i,t}$ are natural and are detailed in Appendix~\ref{Appendix_proofs_thm_qRep}. We also outline a popular algorithm `REINFORCE' (see Appendix \ref{sec: reinforce}), which has been shown to satisfy these assumptions (\cite{giannou2021convergence}).

For the moment, it suffices to state that the stepsizes $\gamma_t$ are fully determined by a parameter $\gamma$ and form a decreasing sequence. While our main results are stated for a common stepsize sequence across all agents, they also extend to more general, agent-dependent stepsize schedules, as explained in the appendix.

The $q$-replicator dynamic updates strategies in a direction of estimated reward improvement. Since the $\eps$-implementation strategy profiles are pure, some histories occur with zero probability. 

Two strategies that differ only in their actions following such zero-probability histories will yield identical realised rewards. Consequently, if the $q$-replicator dynamic reaches a strategy profile that deviates from an $\eps$-implementation solely after zero-probability histories, the dynamic does not move closer to the $\eps$-implementation. This observation motivates the following definition:

\begin{definition}
    Two strategy profiles $\pi$ and $\pi'$ are said to be \emph{equivalent} (denoted $\pi \sim \pi'$) if they share the same set of zero-probability histories and induce the same distribution over actions otherwise. We denote by $\Psi(\pi)$ the equivalence class of $\pi$ under this relation.
\end{definition}

Now to our main result in this section:

\begin{restatable}{theorem}{thmPerfMonqRep}\label{thm: perfmon_qrepl}
 Take $\eps>0$. For every $u\in \tilde{W}$, pick an $\eps$-finite implementation $\pi^*$.
 \begin{enumerate}
     \item There is a neighbourhood $\c U$ of $\pi^*$ such that for every $\eta>0$ and any $\gamma>0$ small enough we have, when $\pi_0\in \c U$,
$$
    \mathbb{P}\left( \pi_t \rightarrow \Psi(\pi^*) \text{ as } t \rightarrow \infty\right )\geq 1-\eta.
$$
\item With probability at least $1-\eta$, the vector of expected discounted payoffs accrued by the $q$-replicator dynamic is $\epsilon$-close to $u$.
 \end{enumerate}
\end{restatable}

This result shows that, under an appropriately chosen recall length and with perfect monitoring, the $q$-replicator dynamic can, with arbitrarily high probability, approximate any individually rational and feasible payoff vector. Hence, the dynamic is capable of learning the full set of outcomes supported by the Folk Theorem.

\section{Imperfect Monitoring}\label{sec: Impfmon}
% adjust discussion?? say all this after q-repl, as catchall for both methods? 

Here, we relax the assumption of perfect monitoring and consider how our previous results for $Q$-learning and $q$-replicator dynamics extend to this setting. Instead of assuming that players perfectly observe their opponents' actions, we now assume that at the end of each period, each player observes a possibly non-deterministic signal that depends on the realised action profile, but does not reveal the specific actions taken during that period.

Imperfect monitoring can be categorised into two main types: public monitoring and private monitoring. In the case of public monitoring, the signal is publicly observable and identical for all players. Conversely, in private monitoring, each player privately observes an individual signal.

When studying public monitoring, it is common practice to restrict players to conditioning their actions solely on the history of observed public signals, rather than on their privately known actions taken. When strategies are conditioned exclusively on public history, the solution concept typically considered is Perfect Public Equilibria (PPE).\footnote{There always exists such a strategy profile. For further readings see \cite{fudenberg2009folk}}

The analysis of PPE has traditionally utilised dynamic programming methods, as shown in works such as, \cite{abreu1990toward}, \cite{fudenberg2007perfect} and \cite{fudenberg2009folk}. This line of research has highlighted various monitoring conditions necessary to achieve feasible payoffs surpassing either the minmax level or the one-shot Nash equilibrium level through PPE. Our results extend only if equilibria constructed in this manner satisfy strictness and finite recall.

\cite{mailath2002repeated} explore PPE that employ a punishment mechanism following any deviation (grim-trigger). For these equilibria, being strict and reliant on finite recall, it is clear that Theorems \ref{thm: perfmon_Qlearning} and \ref{thm: perfmon_qrepl} hold, as one can just reinterpret what were perfect monitoring states now as public history states. 

% However, as discussed in subsection XXX, the fact that \cite{mailath2002repeated} construct equilibria that can get absorbed in punishment states forever, Theorem X may not extend.

When monitoring is public, allowing the players to condition their actions on privately observed own actions (in addition to the publicly observed signals), increases the set of equilibria payoffs that can be obtained (see, \cite{kandori2006efficiency}). However, the common knowledge of the relevant parts of the history is lost. A similar situation arises with private monitoring, where each player observes a private signal. The type of equilibrium that is used in these cases is \emph{sequential equilibrium} (\cite{kreps1982sequential}).
Sequential equilibrium requires tracking an infinite hierarchy of beliefs, when each player updates their belief in a Bayesian manner after each period, given the action played and signal observed. This Bayesian updating accumulates over time, contrasting with finite recall, which requires players to `forget' their observations after a certain number of periods.

Therefore, it is unsurprising that Sequential Equilibria cannot typically be represented by finite machines. Hence, RL schemes will commonly not be able to learn such strategy profiles. Specifically, there exist Sequential Equilibria for which we cannot demonstrate positive probability of convergence, and there are strategy profiles with positive probability of convergence that do not constitute Sequential Equilibria (for a counterexample, see Appendix \ref{sec: sequentialequi}). In this scenario, the set of strategy profiles that we prove can be learnt is the set of strategy profiles $\pi^*$ such that any unilateral deviation of player~$i$ to an alternative strategy within their strategy set induces a strict loss. In other words, each player is playing a best response from the set of strategies they are allowed to use.   

%look into behavioural solution concepts that may fit this. Berk Nash? 

\section{Remarks}

\textbf{On the Folk Theorem.} The wide range of equilibrium payoffs described by the Folk Theorem is sometimes viewed as a drawback, as it diminishes the predictive power of the model in determining the outcome of a game.  However, a different perspective can be taken. The fact that a strategy profile constitutes an equilibrium implies a certain degree of stability.  Thus, the existence of multiple stable strategy profiles suggests that if players are playing a non-Pareto optimal equilibrium, there is a stable strategy profile that could be discovered and adopted with higher payoffs for each of the players.

\textbf{From Finite Recall to Finite Memory.} A careful reading of our proofs reveals that, instead of having players condition their actions upon private histories with finite recall, one can generalise to having players condition their actions upon any private state from a finite set of states. This means that Theorems \ref{thm: perfmon_Qlearning} and \ref{thm: perfmon_qrepl} can be stated in more general terms - with finite memory rather than finite recall. With finite memory it is possible to represent strategies that depend on signals received arbitrarily far in the past, which is not possible with finite recall. The set of strategy profiles thus supported in this regime is substantially different.

\textbf{Consequences for Stochastic Games.} We prove Theorem \ref{thm: perfmon_qrepl} extending the methodology of \cite{giannou2021convergence}. A careful reading of our proof reveals that, when applied to the Stochastic Games framework studied in \cite{giannou2021convergence},  our extension ensures local convergence to strict equilibria for a larger class of dynamics than the projected gradient dynamics that was the sole focus of attention of \cite{giannou2021convergence}.

\section{Directions for future work}\label{sec6_discuss}

%\subsection{Solution Concepts}
%This work has considered solution concepts under various monitoring structures and where each player may have a different finite recall length. Certain instances are well understood; for example, where there is perfect monitoring and all players have the same history length, an $\ell$-recall equilibrium is also an equilibrium of the infinitely repeated game. 

%Generally, however, the $\ell$-recall equilibrium is not, to our knowledge, a subset of any established solution concept in game theory, despite it being a natural solution concept for RL. 
%We believe that future research in analysing convergence points, as well as long-run non-convergent behaviours of RL, is a fruitful direction for future game theoretical research.

\subsection{Behavioural Strategies}

In this work, we have considered each player's strategy to be a probability distribution over their pure strategies. Hence, for player $i \in N$, a strategy has been $\pi_i \in \Delta \left (A_i^{\hat{H}_i^{\ell_i}} \right )$. However, an alternative is to consider behavioural strategies, which are functions that map a player's private history to a distribution over their actions. Explicitly, for player $i \in N$, a behavioural strategy would be of the form $\pi_i : \hat{H}_i^{\ell_i} \to \Delta(A_i)$. In a follow-up work, we will define a version of strict subgame-perfect equilibrium and demonstrate that the dynamics, under these changes, converge locally to such an equilibrium.

\subsection{Coordination of parameters}
Theorems \ref{thm: perfmon_Qlearning} and \ref{thm: perfmon_qrepl} are significant in that they do not require players to have identical stepsizes or, in the case of the $q$-replicator, to use the same estimator for their $q$-gradient. However, further generalisation might be possible. For instance, we could consider a model in which each player may use a different value of $q$. We conjecture that this generalisation would yield similar results to Theorem \ref{thm: perfmon_qrepl}. Furthermore, although we do not assume players use identical stepsizes, we do assume their stepsizes are of the same magnitude. This assumption might also be subject to generalisation in future work.

\subsection{Basins of Attraction and Convergence Rates}

Theorems \ref{thm: perfmon_Qlearning} and \ref{thm: perfmon_qrepl} concern local convergence, which, as previously discussed, is the optimal outcome achievable outside of a limited class of games. However, further research is warranted in this area.  One research direction involves analysing the basin of attraction for locally attracting fixed points and examining the geometrical and topological attributes of these basins. Furthermore, for a game with multiple equilibria, like a repeated game, the relative size of a basin of attraction of an equilibrium can be interpreted as the likelihood of the learning dynamics to converge to this equilibrium. This, in turn, can be considered as a selection mechanism. Understanding the factors that affect the size of the basins of attraction is interesting in this regard as well.  

Another important aspect is studying convergence rates within these basins of attraction. While establishing a bound for the convergence rate is beyond the scope of this paper, existing bounds for similar works (see \cite{giannou2021convergence}) could be insightful, though these results rely on assumptions not directly applicable to repeated games.

\section{Conclusion}

In this paper, we have established that a wide class of reinforcement learning dynamics, including $Q$-learning and $q$-replicator processes, converge locally to strict equilibria in repeated games with finite recall when all the agents follow the same learning method independently. Our results extend the classical Folk Theorem into the realm of algorithmic learning, demonstrating that AI agents can, with high probability, achieve any feasible and strictly individually rational payoff vector. As the families of algorithms we consider encompass substantial variations in updating schemes, our results suggest that convergence to Folk Theorem outcomes is not an artifact of a specific update rule. Rather, the geometry of repeated games with finite recall naturally creates basins of attraction around these strict equilibria, which standard learning algorithms are prone to discover.

Our findings have several implications. First, they provide a rigorous foundation for understanding the conditions under which algorithmic collusion may emerge in strategic settings. Second, they clarify the relationship between classical equilibrium concepts, imperfect monitoring and the behaviour of learning agents, offering a framework to assess which equilibria are practically attainable through reinforcement learning. Finally, our work opens multiple avenues for future research, including the study of convergence rates, basins of attraction, heterogeneous learning parameters, and the extension to behavioural strategies.

Overall, this paper contributes to the understanding of strategic interactions among adaptive agents, demonstrating that the rich set of outcomes predicted by the Folk Theorem is not only theoretically attainable but also practically learnable by AI algorithms.  This includes collusive outcomes in competitive markets: we establish that independent learning agents can converge to supracompetitive prices without communication or explicit coordination, making algorithmic collusion an inherent feature of multi-agent reinforcement learning.

\appendix

\bibliographystyle{apalike.bst}
\bibliography{lit.bib}

\section{Proofs}\label{sec: proofs}

\subsection{Proof of Theorem~\ref{thm: payoff_approx}}

Recall the following theorem that allows us to approximate feasible and strictly individually rational payoffs by strict $\ell$-recall subgame-perfect equilibrium: 

\thmpayoffapprox*

\begin{proof}[Proof of Theorem~\ref{thm: payoff_approx}]
If $\tilde W=\varnothing$, there is nothing to prove. Otherwise, fix an arbitrary payoff vector $u\in\tilde W$ and $\epsilon>0$. Since $\tilde W\neq\varnothing$, the hypotheses of Theorem~1 in \cite{barlo2016bounded} are satisfied, independently of the number of players $|N|$. That result implies that, for $\delta$ sufficiently close to one, there exists a finite recall length $M\in\mathbb{N}$ and an $M$-recall subgame-perfect equilibrium of the $\delta$-discounted repeated game whose induced expected payoff vector lies within $\epsilon$ of $u$.

The equilibrium constructed in \cite{barlo2016bounded} has the familiar Folk Theorem structure: an equilibrium path implementing the target payoff $u$, a punishment phase following unilateral deviations, and a continuation phase thereafter. Although the equilibrium is not strict as stated, it can be made strict by a simple modification. Specifically, extending the punishment phase by one additional period renders any profitable deviation strictly suboptimal, while leaving both the equilibrium path and the induced payoff vector unchanged.\footnote{The only case in which strictness may fail in every subgame is when a punished player has multiple best responses during the punishment phase. By construction, the player may select any of these best responses. This does not affect the argument, since any deviation from a mixture over these best responses yields a strictly lower payoff.}

This modification increases the required recall length by at most $|N|+5$, and hence by a finite amount. Consequently, there exists some $\ell\in\mathbb{N}$ and an $\ell$-recall strict equilibrium $\pi^*$ of $\Gamma(\delta)$ whose induced payoff vector lies within $\epsilon$ of $u$.
\end{proof}

\subsection{Proof of Theorem \ref{thm: perfmon_Qlearning}} \label{Appendix_proofs_thm_Qlearning}

For each player $i \in N$, $Q$-learning is an iteration over a $Q$-matrix $Q_i: H \times A_i \to \mathbb{R}$ representing state-action values. Given some initial $Q_i^0$, at each time $t=1,\dots$, each player $i$ updates their $Q$-matrix $Q_i$ in the following manner:
\begin{align*}
    Q_i^{t+1}(s,a) = Q_i^t(s,a) + \gamma_{n_t(s,a)}\mathbf{1}\left\{s_t=s, a_t=a\right\}\left(r_i(s_t,a_t) + \delta \max_{a'\in A_i}Q_i^t(s_{t+1},a') - Q_i^t(s,a) \right).
\end{align*}
where $n_t(s,a)$ counts the number of visits to $s,a$ up to time $t$ and $\gamma_{n_t(s,a)} = \frac{\gamma}{n_t(s,a)}$ for some $\gamma \in (0,1)$. 

We recall the main convergence result for the $Q$-Learning dynamics.

\thmPerfMonQLearning*

In order to prove this result, we first obtain the following:

\begin{theorem}\label{thm:QLearn_Local}
    Let $\pi^*$ be an $\ell$-recall strict subgame-perfect equilibrium of $\Gamma(\delta)$. Denote by $Q^*$ the $Q$-matrix of $\pi^*$. Then for every $\varepsilon > 0$ and $\eta > 0$, there is a neighbourhood $\mathcal{U}$ of $Q^*$ such that for any $Q^0\in \mathcal{U}$ and any small enough $\gamma>0$ we have, for all times $t>0$,
    \begin{align*}
        \mathbb{P}\left( \pi(Q^t)=\pi(Q^*)\right )\geq 1-\eta.
    \end{align*}
Moreover, with probability at least $1-\eta$, we have that the vector of expected discounted payoffs under the $Q$-learning dynamics is $\varepsilon$-close to the vector of discounted payoffs under $\pi^*$.
\end{theorem}

\begin{proof}[Proof of Theorem~\ref{thm:QLearn_Local}] Fix $\varepsilon, \eta > 0$. $\pi^*$ is strict, therefore let 
\begin{align*}
\xi \;=\; \min_i\min_{s\in S}\biggl(\max_{a\in A_i} Q^*(s,a) - \max_{a'\in A_i\setminus\{a\}} Q^*(s,a')\biggr) \;>\;0.
\end{align*}
Let $\xi'=\xi/4$. Let $\mathcal{U} := \{ Q : \lVert Q - Q^* \rVert_\infty < \xi' \}$ and fix $Q^0 \in \mathcal{U}$. Fix $L \in \mathbb{N}$ large enough so that:
\begin{align*}
L &\geq \frac{M}{2 \xi' (1- \delta)}, \\
\frac{2M\,\delta^L}{1-\delta} &< \varepsilon,
\end{align*}
where $M$ is the uniform bound on elements in the $Q$-table. Fix $\gamma > 0$ small enough such that 
\begin{align*}
    |S| |N|^2\gamma^2 \left( (1 + \log(2L))^2  + \frac{\pi^2}{3} \right) + \frac{|N|^2 \gamma^2 \pi^2}{12} + |S|L |N|\gamma < \eta.
\end{align*}
    
$Q$-learning starting at $Q^0$ gives rise to a random sequence of states that we denote by $s_1, s_2, s_3, \dots{}$. For any fixed state $s$, let:
\begin{itemize}
    \item $T_s$ be the set of times $t$ such that $s_t=s$.
    \item $T^t_s=T_s\cap \{1,\dots,t\}$ for some $t$.
    \item $T^t_s(L)$ be the $L$-many largest elements of $T^t_s$.
\end{itemize}

Furthermore, let $T^t_{s,a_i}$ be the times $\tau$ such that $s_\tau =s$ and player $i$ plays action $a_i$. These are the times where the $Q$-value for player $i$ of $(s,a_i)$ changes.  

On the probability space induced by this process, the number of deviations at times $t \geq 1$ is a sequence of random variables $(d^t)_{t\ge 1}$ with $d^t \sim \operatorname{Bin}\bigl(|N|,\gamma/t\bigr)$. We denote the set of bad events (with respect to the same probability space as $d^t$) as:
\begin{align*}
    \mathcal{B}(L) = \Bigg \lbrace \text{events} : & \exists t \text{ such that }  \\
    & \left(\left(\sum_{\tau\in T^t_{s_t}(L)} d^\tau \geq 2\right) \text{ or }\left(\sum_{\tau\in T^t_{s_t}(L)} d^\tau \geq 1\text{ and }0<|T^t_{s_t}(L)|< L\right)\right) \Bigg \rbrace.
\end{align*}

So, $\mathcal{B}(L)$ is the event of ``bad exploration samples'' for which there exists some time $t$ such that either:
\begin{itemize}
    \item In the last $L$ times that state $s_t$ has been visited (including time $t$), at least two deviations have occurred.  
    \item That state $s_t$ has been visited fewer than $L$ times (including time $t$), and a deviation has occurred. 
\end{itemize}

The idea of the proof is as follows:
\begin{enumerate}
    \item With $\gamma$ as above, we have that $\mathbb{P}[\mathcal{B}(L)]< \eta$.
    \item With $L$ as above, we have, conditioning on $\mathcal{B}(L)$ not occurring, that $\pi(Q^t) = \pi(Q^*)$ for all $t \geq 0$.
\item  With $L$ as above, the vector of discounted payoffs from the beginning of the learning is $\varepsilon$-close to the vector of discounted payoffs obtained playing $\pi^*$. 
\end{enumerate}

\textbf{Part One:} Show that for every $L > 0$, we can choose $\gamma > 0$ small enough such that $\mathbb{P}[\mathcal{B}(L)]< \eta$.

For each state $s$, consider its visit sequence (times $\tau_1^s < \tau_2^s < \tau_3^s < \ldots$). We first try to bound the probability that, within $L$ visits to state $s$ there are two deviations at two different times. To do this, we consider overlapping windows of $2L$ visits, and bound the probability that there are two deviations at two different times within any of these windows. For the $m$-th window (visits $(m-1)L+1$ through $(m+1)L$), these occur at times $\tau_{(m-1)L+1}^s, \ldots, \tau_{(m+1)L}^s$. The probability that there are two deviations at two different times is:
\begin{align*}
\sum_{(m-1)L+1 \leq i < j \leq  (m+1)L} |N|^2 \frac{\gamma}{\tau^s_i} \cdot \frac{\gamma}{\tau^s_j}
\end{align*}

Since $\tau_j^s \geq j$ for all $j$, we upper bound by:
\begin{align*}
|N|^2 \gamma^2 \sum_{(m-1)L+1 \leq i < j \leq  (m+1)L} \frac{1}{i \cdot j} 
\end{align*}

We now bound this sum for each window:

\textbf{For $m = 1$:} We have $1 \leq i < j \leq 2L$, so:
\begin{align*}
\sum_{1 \leq i < j \leq 2L} \frac{1}{i \cdot j} &\leq \left(\sum_{k=1}^{2L} \frac{1}{k}\right)^2 \leq (1 + \log(2L))^2 
\end{align*}

\textbf{For $m \geq 2$:} We have $(m-1)L+1 \leq i < j \leq (m+1)L$. Since all indices are at least $(m-1)L$:
\begin{align*}
\frac{1}{i \cdot j} \leq \frac{1}{((m-1)L)^2} = \frac{1}{(m-1)^2L^2}
\end{align*}

The number of pairs $(i,j)$ in the window is $\binom{2L}{2} < 2L^2$, so:
\begin{align*}
\sum_{(m-1)L+1 \leq i < j \leq (m+1)L} \frac{1}{i \cdot j} < 2L^2 \cdot \frac{1}{(m-1)^2L^2} = \frac{2}{(m-1)^2}
\end{align*}

Therefore, the probability that there are two deviations at two different times within $L$ visits to any state $s$ is upper bounded by:
\begin{align*}
&|S||N|^2 \gamma^2 \sum_{m=1}^\infty \sum_{(m-1)L+1 \leq i < j \leq  (m+1)L} \frac{1}{i \cdot j} \nonumber \\
&\leq |S| |N|^2\gamma^2 \left( (1 + \log(2L))^2 + \sum_{m=2}^\infty \frac{2}{(m-1)^2} \right) \nonumber \\
&= |S||N|^2 \gamma^2 \left( (1 + \log(2L))^2  + 2 \sum_{k=1}^\infty \frac{1}{k^2} \right) \nonumber \\
&= |S| |N|^2\gamma^2 \left( (1 + \log(2L))^2  + 2 \cdot \frac{\pi^2}{6} \right) \nonumber \\
&= |S| |N|^2\gamma^2 \left( (1 + \log(2L))^2  + \frac{\pi^2}{3} \right)
\end{align*}
The next event that we need to consider is that, at some time $t$, we have $d^t \geq 2$. For this, we have:
\begin{align*}
    \sum_{t=1}^\infty \mathbb{P}(d^t \geq 2) & \leq  \sum_{t=1}^\infty \frac{|N|^2}{2}\cdot\frac{\gamma^2}{t^2} \\
    & = \frac{|N|^2 \gamma^2 \pi^2}{12}
\end{align*}
Hence, 
\begin{align*}
    \mathbb{P} \left ( \exists t  : \sum_{\tau\in T^t_{s_t}(L)} d^\tau \geq 2\right) \leq |S| |N|^2\gamma^2 \left( (1 + \log(2L))^2  + \frac{\pi^2}{3} \right) + \frac{|N|^2 \gamma^2 \pi^2}{12}
\end{align*}

We also bound:
\begin{align*}
    \mathbb{P} \left ( \exists t : \sum_{\tau\in T^t_{s_t}(L)} d^\tau \geq 1\text{ and }0<|T^t_{s_t}(L)|< L \right) & \leq \sum_s \sum_{t=1}^{L-1} \mathbb{P}(d^t \geq 1) \\
    & \leq |S|L |N|\gamma
\end{align*}
Hence, 
\begin{align*}
    \mathbb{P}[\mathcal{B}(L)] \leq |S| |N|^2\gamma^2 \left( (1 + \log(2L))^2  + \frac{\pi^2}{3} \right) + \frac{|N|^2 \gamma^2 \pi^2}{12} + |S|L |N|\gamma,
\end{align*}

and so, by the assumption on $\gamma$, we have $\mathbb{P}[\mathcal{B}(L)]<\eta$, which proves Part One.

\textbf{Part Two:} 
We now condition on the event $\mathcal{B}(L)^c$. Suppose for contradiction that there exists a first time $\bar{t}$ and a state $s'=s_{\bar{t}}$, player $i'$ and action $a'$ for which
\begin{align*}
\bigl|Q^{\,\bar{t}+1}_{i'}(s',a') - Q^*_{i'}(s',a')\bigr| \;>\; 2\xi'.
\end{align*}
Since $\bar{t}$ is the first such time, we have $\|Q_{i'}^t - Q_{i'}^*\|_\infty \leq 2\xi'$ for all $t \leq \bar{t}$. Suppose that all non-$i'$ players do not deviate at time $\bar{t}$. Then,
\begin{align*}
Q_{i'}^{\bar{t}+1}(s',a') &= (1- \alpha^{i'}_{n_{\bar{t}}^{i'}(s',a')})Q_{i'}^{\,\bar{t}}(s',a') \\
& \quad + \alpha^{i'}_{n_{\bar{t}}^{i'}(s',a')} (r_{i'}(s',(a',a_{-i}^*(s'))) + \delta\max_{b} Q^{\,\bar{t}}_{i'}(f(s',(a',a_{-i}^*(s'))),b)).
\end{align*}
As $\pi^*$ is an SPNE, we have that 
\begin{align*}
Q_{i'}^*(s',a') = r_{i'}(s',(a',a_{-i}^*(s'))) + \delta\max_{b} Q^*_{i'}(f(s',(a',a_{-i}^*(s'))),b).
\end{align*}
Therefore,
\begin{align*}
    Q_{i'}^{\bar{t}+1}(s',a') - Q_{i'}^*(s',a') &= (1- \alpha^{i'}_{n_{\bar{t}}^{i'}(s',a')})(Q_{i'}^{\,\bar{t}}(s',a')- Q_{i'}^*(s',a')) \\
& \quad + \alpha^{i'}_{n_{\bar{t}}^{i'}(s',a')} \delta ( \max_{b} Q^{\,\bar{t}}_{i'}(f(s',(a',a_{-i}^*(s'))),b)- \max_{b} Q^{*}_{i'}(f(s',(a',a_{-i}^*(s'))),b)).
\end{align*}
By the triangle inequality and minimality of $\bar{t}$, we therefore have:
\begin{align*}
   | Q_{i'}^{\bar{t}+1}(s',a') - Q_{i'}^*(s',a') | &\leq  2 \xi'(1- \alpha^{i'}_{n_{\bar{t}}^{i'}(s',a')}) +  \alpha^{i'}_{n_{\bar{t}}^{i'}(s',a')} \delta (2 \xi') \\
   & \leq 2 \xi',
\end{align*}
which is a contradiction. Therefore, at time $\bar{t}$ there must be a non-$i'$ player who deviates. As we condition on the event $\mathcal{B}(L)^c$, it must be that $s'$ has been visited at least $L$ times (including time $\bar{t}$) and this non-$i'$ player deviation is the only deviation out of all times in $T_{s_{\bar{t}}}^{\bar{t}}(L)$. Note that this implies that $a'$ is played by player $i'$ for all times $t \in T_{s_{\bar{t}}}^{\bar{t}}(L)$. For every time $t \in T^{\bar{t}}_{s_{\bar{t}}}(L)$ except $t = \bar{t}$, as there are no deviations occurring, we denote the common successor state by $s^+ := s_{t+1}$ and the Bellman targets by
\begin{align*}
T_{t} := r_{i'}(s',(a',a_{-i}^t)) + \delta\max_{b} Q^{\,t}_{i'}(s^+,b).
\end{align*}
By minimality of $\bar{t}$, we have, for these times, that $T_{t} = Q_i^*(s',a') \pm 2\delta\xi'$. For time $\bar{t}$, the Bellman target is 
\begin{align*}
T_{\bar{t}} := r_{i'}(s',(a',a_{-i}^{\bar{t}})) + \delta\max_{b} Q^{\,\bar{t}}_{i'}(s^{\bar{t}+1},b),
\end{align*}
which is a priori only bounded by $M$. 

The single-step $Q$ update for $(s',a')$ at a visit time $t$ is
\begin{align*}
Q_{i'}^{t+1}(s',a') = (1- \alpha^{i'}_{n_t^{i'}(s',a')})Q_{i'}^{\,t}(s',a') + \alpha^{i'}_{n_t^{i'}(s',a')} T_{t}.
\end{align*}
Let $\tau = \min T^{\bar{t}}_{s_{\bar{t}}}(L)$. We unwind this recursion to yield
\begin{align*}
Q_{i'}^{\bar{t}+1}(s',a') = w_0\,Q_{i'}^{\tau}(s',a') + \sum_{t \in T^{\bar{t}}_{s_{\bar{t}}}(L)} w_t T_{t},
\end{align*}
where
\begin{align*}
w_0 &= \prod_{t \in T^{\bar{t}}_{s_{\bar{t}}}(L)} \left(1 - \alpha_{n^{i'}_{t}(s',a')}\right), \\
w_t &= \alpha_{n^{i'}_{t}(s',a')} \prod_{l \in T^{\bar{t}}_{s_{\bar{t}}}(L), l > t} \left(1 - \alpha_{n^{i'}_{l}(s',a')}\right) \quad \text{for } t \in T^{\bar{t}}_{s_{\bar{t}}}(L).
\end{align*}

Note that $w_0 + \sum_{t \in T^{\bar{t}}_{s_{\bar{t}}}(L)} w_t = 1$ and all weights are non-negative. Subtracting $Q_{i'}^*(s',a')$ and applying the triangle inequality gives
\begin{align*}
\bigl|Q_{i'}^{\bar{t}+1}(s',a') - Q_{i'}^*(s',a')\bigr|
&\le w_0\bigl|Q_{i'}^{\tau}(s',a')- Q_{i'}^*(s',a') \bigr| + \sum_{t \in T^{\bar{t}}_{s_{\bar{t}}}(L)} w_t \bigl|T_{t}- Q_{i'}^*(s',a')\bigr|. 
\end{align*}
Let $m = n^{i'}_{\tau}(s',a') - 1$ be the number of times $(s',a')$ was visited before time $\tau = \min T^{\bar{t}}_{s_{\bar{t}}}(L)$. In the above, we obtain weights that depend on the consecutive visit counts to this state-action pair. Specifically, these updates use stepsizes $\frac{\alpha}{m+1}, \frac{\alpha}{m+2}, ..., \frac{\alpha}{m+L}$. Ordering the elements of $T^{\bar{t}}_{s_{\bar{t}}}(L)$ chronologically and reindexing the corresponding weights accordingly, we may write
\begin{align*}
    (w_0, w_t : t \in T^{\bar{t}}_{s_{\bar{t}}}(L)) \equiv (v_0, v_1, \ldots, v_L),
\end{align*}
where
\begin{align*}
v_0 &= \prod_{k=1}^{L} \left(1 - \frac{\alpha}{k+m}\right), \\
v_k &= \frac{\alpha}{k + m} \prod_{l = k+1}^{L} \left(1 - \frac{\alpha}{l+m }\right) \quad \text{for } k \in \{1,...,L\}.
\end{align*}
We therefore have
\begin{align*}
\bigl|Q_{i'}^{\bar{t}+1}(s',a') - Q_{i'}^*(s',a')\bigr|
&\le w_0\bigl|Q_{i'}^{\tau}(s',a')- Q_{i'}^*(s',a')\bigr| + \sum_{t \in T^{\bar{t}}_{s_{\bar{t}}}(L)} w_t \bigl|T_{t}- Q_{i'}^*(s',a')\bigr| \\
&\le v_0\cdot 2\xi' + \sum_{k =1}^{L-1} v_k\cdot 2\delta\xi' + v_{L}\cdot M\\
&= 2\xi' v_0 + 2\delta\xi'(1-v_0-v_{L}) + v_{L} M \\
&\leq 2\xi' (1- \delta)v_0 + 2\delta \xi' + M v_{L}.
\end{align*}
We want
\begin{align*}
2\xi' (1-\delta)v_0 + 2\delta\xi' + M v_{L} \le 2\xi',
\end{align*}
or equivalently
\begin{align*}
M v_{L} \le 2\xi'(1-\delta)(1-v_0).
\end{align*}
As $\alpha < 1$, $v_1 \geq v_2 \geq ... \geq v_{L}$. Therefore, $1-v_0 = \sum_{k=1}^{L} v_k \geq Lv_{L}$. Hence, it suffices to show that:
\begin{align*}
M  \le 2\xi'(1-\delta)L,
\end{align*}
which is given by the condition on $L$ previously provided. 

Therefore no such time $\bar{t}$ exists: for every time $t$, for each player $i$ and every $(s,a)$ we have
\begin{align*}
\bigl|Q_{i'}^{\,t}(s,a) - Q_{i'}^*(s,a)\bigr| \le 2\xi' = \frac{\xi}{2}.
\end{align*}
Therefore, for all players, the argmaxes in each state are the same across $Q_{i'}^t$ as $Q_{i'}^*$. Thus, 
\begin{align*}
    \pi(Q^t) = \pi(Q^*) \text{ for all } t \geq 0.
\end{align*}

\textbf{Part Three:} 

We show that, outside of $\mathcal{B}(L)$, the vector of discounted expected payoffs at the beginning of learning is arbitrarily close to the vector of discounted payoffs obtained by playing $\pi^*$.

Define the two vectors as the discounted expected payoff (from the beginning, stage $t=0$ onward) for each player:
\begin{align*}
V^* &:= \mathbb{E}\Biggl[\sum_{t=0}^\infty \delta^t\, r(s_t,a_t)\;\Big|\; \text{players follow } \pi^*\Biggr],\\
V  &:= \mathbb{E}\Biggl[\sum_{t=0}^\infty \delta^t\, r(s_t,a_t)\;\Big|\; \text{players follow the }Q\text{-learning dynamics}\Biggr].
\end{align*}

As we condition on $\mathcal{B}(L)$ not occurring, no exploration occurs during the first $L$ stages of play. Therefore, the realised action profile at stage $t$ coincides with the action profile induced by $\pi^*$ for every $t=0,1,\dots,L-1$. Hence, the realised stage payoff vectors agree for those first $L$ stages:
\begin{align*}
r(s_t,a_t) = r^*(s_t,a_t)\qquad\text{for }t=0,1,\dots,L-1,
\end{align*}
where $r^*(s_t,a_t)$ denotes the stage payoff under $\pi^*$. Therefore, for each player $i$,
\begin{align*}
\bigl|V_i - V^*_i\bigr|
&= \Biggl|\mathbb{E}\Bigl[\sum_{t=L}^\infty \delta^t\bigl(r_i(s_t,a_t)-r_i^*(s_t,a_t)\bigr)\,\Big|\,\mathcal{B}(L)^c\Bigr]\Biggr| \\
&\le \mathbb{E}\Biggl[\sum_{t=L}^\infty \delta^t\bigl|r_i(s_t,a_t)-r_i^*(s_t,a_t)\bigr|\,\Big|\, \mathcal{B}(L)^c\Biggr] \\
&\le \sum_{t=L}^\infty \delta^t\cdot 2M \\
&= 2M \sum_{t=L}^\infty \delta^t
= 2M\,\delta^L\sum_{k=0}^\infty \delta^k
= \frac{2M\,\delta^L}{1-\delta}.
\end{align*}

Taking the supremum over players gives the uniform bound
\begin{align*}
\|V - V^*\|_\infty \le \frac{2M\,\delta^L}{1-\delta}.
\end{align*}
With $L$ as previously defined we therefore have
\begin{align*}
\|V - V^*\|_\infty < \varepsilon.
\end{align*}

\end{proof}

We now combine Theorem~\ref{thm:QLearn_Local} and Theorem~\ref{thm: payoff_approx} to obtain the main result for $Q$-learning dynamics.

\begin{proof}[Proof of Theorem~\ref{thm: perfmon_Qlearning}]
Fix $\varepsilon>0$. By Theorem~\ref{thm: payoff_approx}, there exists $\delta^*\in(0,1)$ such that for all $\delta\in(\delta^*,1)$ and every $u\in\tilde W$, one can find $\ell\in\mathbb{N}$ and an $\ell$-recall strict equilibrium $\pi^*$ of $\Gamma(\delta)$ such that the distance between $u$ and  $V(\pi^*)$ is at most $\epsilon$. Fix $\delta\in(\delta^*,1)$, an arbitrary $u\in\tilde W$, and a corresponding equilibrium $\pi^*$. We now apply Theorem~\ref{thm:QLearn_Local} to $\pi^*$, which completes the proof.
\end{proof}

\subsection{Characterisation of equilibria}

In this section, we present results that connect variational inequalities involving the $q$-gradient to solution concepts from game theory. Later, we use these variational inequalities to prove our results for the learning dynamics. 
These results hold without change if we allow the players to have different recall lengths, that is, when, with abuse of notation, we denote: $\hat{H}^\ell=\prod_{i\in N}{\hat{H}_i}^{\ell_i}$  for some $\ell \in \mathbb{Z}_{+}^{|N|}$.

\begin{restatable}{lemma}{lemmaRepeatedRestate}\label{lemmaRepeated}
Let $G$ be an $\ell$-recall repeated game as previously defined. For any $q\geq 0$, a strategy profile $\pi^*\in \Pi^\ell$ is a strict Nash equilibrium if and only if the following two conditions are satisfied: 
\begin{itemize}
    \item[C'(i)] For any $\pi\in \Pi^\ell$ we have $\langle v^q(\pi^*), \pi-   \pi^*\rangle \leq 0$.
    \item[C'(ii)] There exists $\epsilon>0$ such that for any $\pi\in \Pi^\ell\setminus \Psi(\pi^*)$ at distance at most $\epsilon$ from $\pi^*$ we have $\langle v^q(\pi), \pi-   \pi^*\rangle < 0$.
\end{itemize}
\end{restatable}

The intuition behind this result is that the variational inequalities can be viewed as a sum of the inner products of $\pi_i - \pi^*_i$ with $v_i(\pi^*)$ or $v_i(\pi)$ for each player $i$. If this inner product is negative, it suggests that moving from $\pi_i$ towards $\pi^*_i$ increases the weight on strategies with higher payoffs, assuming the opponents' strategies are fixed at $\pi^*_{-i}$ or $\pi_{-i}$. Thus, the first condition indicates no profitable deviations from $\pi_i^*$, and the second condition suggests that moving towards $\pi_i^*$ improves the strategy when near $\pi^*$.

Following on from this, for $q=0$, condition \textit{C'(i)} is equivalent to $\pi^*$ being a Nash equilibrium, but, for $q > 0$, it is equivalent to $\pi^*$ being a strategy profile where each player plays their best response among their strategies in the support of $\pi^*$.  

The proof is divided into two propositions; Proposition~\ref{propChar} that proves the lemma for $q=0$, and Proposition~\ref{propCharb} that proves it for $q>0$. 

In the proofs of Propositions~\ref{propChar} and~\ref{propCharb}, we make use of the following definitions and notation.

Let $E_i$ be the set of pure strategies of player $i$, that is, the set of extremal points of $\Pi_i^{\ell_i}$. We index this set according to the order of the components in the vector $\pi_i$, and we denote by $e_\alpha\in E_i$ the pure strategy associated with the $\alpha$-th component of $\pi_i$.

For $e_\alpha\in E_i$, let $\pi_i(e_\alpha)$ denote the $\alpha$-th component of $\pi_i$, which is the probability that the strategy $\pi_i$ assigns to the pure strategy $e_\alpha$. 

\medskip

\begin{proposition}\label{propChar} Lemma \ref{lemmaRepeated} holds in the case $q=0$, which, with the notation of Lemma \ref{lemmaRepeated} is equivalent to the following.
    \begin{enumerate}[label=(\alph*)]
    \item\label{C1partA} Condition \textit{C'(i)} is equivalent to $\pi^*$ being a Nash equilibrium. 
    \item\label{C1partB} The strategy profile $\pi^*$ being a strict equilibrium implies \textit{C'(ii)}.
    \item\label{C1partC} An equilibrium satisfying \textit{C'(ii)} is strict.
    \end{enumerate}
\end{proposition}

\begin{proof}
In the notation, when clear from context, we omit the history length for simplicity.

\noindent\textbf{Proof of \ref{propChar}\ref{C1partA}}
Condition \textit{C'(i)} in the case $q=0$ reads as follows. For every $\pi\in \Pi^\ell$, we have $\langle v^0(\pi^*), \pi-   \pi^*\rangle \leq 0$. In particular, by rearranging the terms and making the sum explicit, we get that condition \textit{C'(i)} holds for $q=0$ if and only if:

$$\sum_{i\in N}\sum_{e_\alpha \in E_i}\left(V_i(e_\alpha,\pi^*_{-i})\pi_i(e_\alpha)\right)\leq \sum_{i\in N}V_i(e^*_{\alpha(i)},\pi^*_{-i}).$$

By linearity of the rewards, this is equivalent to:

$$\sum_{i\in N}V_i(\pi_i,\pi^*_{-i})\leq \sum_{i\in N}V_i(e^*_{\alpha(i)},\pi^*_{-i}).$$

If $\pi^*$ is an equilibrium then (by definition) no player has a profitable deviation, which is for every $i\in N$, for every $\pi_i\in \Pi_i^{\ell_i}$, it holds that $V_i(\pi_i,\pi^*_{-i})\leq V_i(\pi^*)$, so \textit{C'(i)} holds if $\pi^*$ is an equilibrium.

We prove the other direction by contradiction. Assume that there is a $\pi^*\in \Pi^\ell$ that is not an equilibrium, yet for which \textit{C'(i)} holds. The strategy profile $\pi^*$ not being an equilibrium implies that there exists (at least) one player $j\in N$ that has a unilateral profitable deviation, that is, there exists $e_\beta\in E_j$ such that $V_j(e_\beta,\pi^*_{-j})> V_j(\pi^*)$. The strategy profile $\pi=(e_\beta,\pi^*_{-j})$ gives a contradiction to \textit{C'(i)}: 
\begin{align*}
    \sum_{i\in N}V_i(\pi_i,\pi^*_{-i}) &=\sum_{i\in N\setminus\{j\}}V_i(\pi^*)+V_j(e_\beta,\pi^*_{-j})\\
    &>
 \sum_{i\in N\setminus\{j\}}V_i(\pi^*)+V_j(\pi^*)\\&= \sum_{i\in N}V_i(\pi^*).
\end{align*}

We conclude that for $q=0$, \textit{C'(i)} holds if and only if $\pi^*$ is an equilibrium.

\bigskip

\noindent\textbf{Proof of \ref{propChar}\ref{C1partB}}
Suppose $\pi^*$ is a strict equilibrium (and therefore each player plays a pure strategy), let us denote by $\alpha^*(i)$ the index in the vector $\pi^*_i$ of the pure strategy played by player $i$ in $\pi^*$. Moreover, let us denote by $e_{\alpha^*(i)}$ the pure strategy that player $i$ plays according to $\pi^*_i$ (the strategy with index $\alpha^*(i)$).

To prove that \textit{C'(ii)} holds, we need to show that for any $\pi\in \Pi^\ell\setminus \Psi(\pi^*)$ close enough to $\pi^*$ we have $\langle v(\pi), \pi-   \pi^*\rangle < 0$. As in part \textit{(a)}, this is equivalent to showing that for any such $\pi$ the following holds:

\[\sum_{i\in N}V_i(\pi^*_i, \pi_{-i})>\sum_{i\in N}V_i(\pi_i, \pi_{-i}).\]

What makes the proof of this inequality not trivial is that $\pi$ may have several players placing positive probability on deviations from $\pi^*$. Consider for example the case where according to $\pi$, player $i$ and player $j$ have a positive probability of playing outside of $\Psi_i(\pi^*)$ and $\Psi_j(\pi^*)$ respectively. 
Because $\pi^*$ is a strict equilibrium, player $i$ incurs a strict loss from their own deviations when all the other players play according to $\pi^*_{-i}$. However, it is possible that if player $j$ and player $i$ both deviate simultaneously, one (or more) players gain a higher reward than the one induced by $\pi^*$.

The (rather tedious) computations of the bounds use the proximity of $\pi$ to $\pi^*$ to show that such simultaneous deviations are taking place with such a small probability that their influence on the gradient of the expected reward is negligible.

As we mentioned, $\pi$ being close to $\pi^*$ entails that, in $\pi$, every player plays according to $\pi^*$ with high probability. For every player $i$, we denote by $\varepsilon_i$ the probability according to $\pi$ that player $i$ plays outside of $e^*_{\alpha(i)}$. More concisely, $\varepsilon_i=1-\pi_i(e^*_{\alpha(i)})$. Note that the quantification that states the closeness of $\pi$ to $\pi^*$ can be translated to a bound for $\varepsilon_i$.

As we mentioned, the key to this part of the proof is to measure the effect of simultaneous deviations. For this reason, we introduce notation to denote events where zero, one, or more deviations occur at the same time.  We denote by $A$ the set of pure strategy profiles where each player plays in $\Psi(\pi^*)$, which is, $A = \{e \in \prod_{i \in N} E_i : \forall j \in N,\  e_j \in \Psi_j(\pi^*) \}$. We denote by $B_j$ the set of pure strategy profiles where player $j$ plays a pure strategy outside of $\Psi_j(\pi^*)$ while all the other players play according to $\pi^*_{-j}$, which is for a given player $j\in N$ we have $B_j = \{e \in \prod_{i \in N} E_i : e_j \notin  \Psi_j(\pi^*),\text{ and } \forall  k\in N\setminus\{j\},\ e_k =  e^*_{\alpha(k)}\}$. Finally, we denote by $C$ all the other pure strategy profiles, where at least $2$ players deviate from $\Psi(\pi^*)$. We have, $C = (\prod_{i \in N} E_i) \setminus ( A \cup (\bigcup_{j \in N} B_j) )$.

Using the notation we just introduced, we are ready to rewrite $V_i(\pi)$. As $V_i(\pi)$ is the expectation of reward for player $i$ under policy $\pi$, and because of the linearity of expectations, we can split $V_i(\pi)$ as a sum according to $A, B_j, C$. To better understand the calculations that follow, it is sometimes useful to remember that $\pi$ is a vector (needed for example when thinking about $\langle v(\pi), \pi - \pi^* \rangle$)  but also a distribution of a random variable that might assume variables in $A, B_j$ or $C$ (which comes in handy when dealing with $V_i(\pi) = \mathbb{E}_{\tau \sim \pi}[V_i(\tau)]$). This also means that the pure strategy profiles in $A, B_j$ or $C$ can also be seen as events for the probability distribution $\pi$ where the realised sample $x$ from $\pi$ either has no, one, or more deviations respectively.

\begin{alignat*}{3}
V_i(\pi) &
\begin{aligned}[t]
       &= \mathbb{E}_{x\sim \pi}[V_i(x) | x\in A ]\cdot \mathbb{P}(A) &&+ \sum_{j \in N} \mathbb{E}_{x\sim \pi}[V_i(x) | x\in B_j ] \cdot \mathbb{P}( B_j ) \\ 
    & &&+ \mathbb{E}_{x\sim \pi}[V_i(x) | x\in C] \cdot \mathbb{P}(C) \\  
\end{aligned}\\
    & = V_i(\pi^*) \prod_{j \in N} \left (1 - \sum_{e_\alpha \in  E_j \setminus \Psi_j(\pi^*)} \pi_j(e_\alpha) \right ) \\
    & \begin{aligned}[t]
    &+ \sum_{j \in N} \sum_{e_\alpha \in E_j \setminus \Psi_j(\pi^*)} &&V_i(e_\alpha , \pi^*_{-j}) \frac{\pi_j(e_\alpha)}{\sum\limits_{e_\alpha \in  E_j \setminus \Psi_j(\pi^*)} \pi_j(e_\alpha) } \left (\sum_{e_\alpha \in  E_j \setminus \Psi_j(\pi^*)} \pi_j(e_\alpha) \right ) \prod_{k \neq j} (1- \varepsilon_k) \\
    & &&+ O\left (\sum_{j \neq k} \varepsilon_j \varepsilon_k \right ) \\
    \end{aligned}\\
    & \begin{aligned} 
    &= V_i(\pi^*) \prod_{j \in N} \Big(1 - && \sum_{e_\alpha \in  E_j \setminus \Psi_j(\pi^*)} \pi_j(e_\alpha) \Big) 
    + \sum_{j \in N} \sum_{e_\alpha \in E_j \setminus \Psi_j(\pi^*)} V_i(e_\alpha , \pi^*_{-j}) \pi_j(e_\alpha) \prod_{k \neq j} (1- \varepsilon_k) \\
    & &&+ O\left (\sum_{j \neq k} \varepsilon_j \varepsilon_k \right ).
    \end{aligned}\\
\end{alignat*}

Note that for any player $j$ in $N$, for any pure strategy $e_\alpha$  in $E_j \setminus \Psi_j(\pi^*)$, we have that $\pi_j(e_\alpha) <\varepsilon_j$. Hence, 
\begin{align*}
    V_i(\pi) & \begin{aligned}[t]
        = V_i(\pi^*) \left (1 - \sum_{j \in N} \sum_{e_\alpha \in  E_j \setminus \Psi_j(\pi^*)} \pi_j(e_\alpha) \right ) 
    &+ \sum_{j \in N} \sum_{e_\alpha \in E_j \setminus \Psi_j(\pi^*)} V_i(e_\alpha , \pi^*_{-j}) \pi_j(e_\alpha)\\  &+ O\left (\sum_{j \neq k} \varepsilon_j \varepsilon_k \right ) 
    \end{aligned}\\
    & = V_i(\pi^*) + \sum_{j \in N} \sum_{e_\alpha \in E_j \setminus \Psi_j(\pi^*)} ( V_i(e_\alpha , \pi^*_{-j}) - V_i(\pi^*) ) \pi_j(e_\alpha)  + O\left (\sum_{j \neq k} \varepsilon_j \varepsilon_k \right ).
\end{align*}
Analogously, 
\begin{equation*}
     V_i(\pi^*_i, \pi_{-i}) = V_i(\pi^*) + \sum_{j \neq i} \sum_{e_\alpha \in E_j \setminus \Psi_j(\pi^*)} ( V_i(e_\alpha , \pi^*_{-j}) - V_i(\pi^*) ) \pi_j(e_\alpha)  + O\left (\sum_{j \neq k} \varepsilon_j \varepsilon_k \right ).
\end{equation*}
Hence, 
\begin{align*}
    V_i(\pi^*_i, \pi_{-i}) - V_i(\pi) &= \begin{aligned}[t]
     V_i(\pi^*) &+ \sum_{j \neq i} \sum_{e_\alpha \in E_j \setminus \Psi_j(\pi^*)} ( V_i(e_\alpha , \pi^*_{-j})  - V_i(\pi^*) ) \pi_j(e_\alpha) - V_i(\pi^*) \\
    & - \sum_{j \in N} \sum_{e_\alpha \in E_j \setminus \Psi_j(\pi^*)} ( V_i(e_\alpha , \pi^*_{-j}) - V_i(\pi^*) ) \pi_j(e_\alpha) \\
    &+ O\left (\sum_{j \neq k} \varepsilon_j \varepsilon_k \right )
    \end{aligned}\\
    &= - \sum_{e_\alpha \in E_i \setminus \Psi_i(\pi^*)} ( V_i(e_\alpha , \pi^*_{-i})  - V_i(\pi^*) ) \pi_i(e_\alpha)  + O\left (\sum_{j \neq k} \varepsilon_j \varepsilon_k \right ) \\
    & > 0.
\end{align*}

\noindent\textbf{Proof of \ref{propChar}\ref{C1partC}}
The proof relies on similar ideas as those of \ref{propChar}\ref{C1partA}. Suppose the equilibrium satisfies \textit{C'(ii)}. We show that $\pi^*$ is strict. Consider a unilateral deviation of player~$i$ to $\pi_i \notin \Psi_i(\pi^*)$. For $\varepsilon \in (0,1)$, define the perturbed profile $\pi'(\varepsilon)$ where player $i$ plays $\pi_i$ with probability $\epsilon$, and $\pi^*$ with remaining probability, while the other players always play $\pi^*_{-i}$, i.e. $\pi'(\varepsilon) \;=\; \bigl(\varepsilon\,\pi_i + (1-\varepsilon)\,\pi^*_i,\; \pi^*_{-i}\bigr)$. For all sufficiently small $\varepsilon > 0$, the profile $\pi'(\varepsilon)$ lies within the neighbourhood prescribed by \textit{C'(ii)} and satisfies $\pi'(\varepsilon) \notin \Psi(\pi^*)$, since $\pi_i \notin \Psi_i(\pi^*)$. From \textit{C'(ii)},
$ \bigl\langle v^0\bigl(\pi'(\varepsilon)\bigr),\; \pi'(\varepsilon) - \pi^* \bigr\rangle < 0$. By the equivalence established in part~\ref{C1partA}, this is equivalent to
\begin{align*}
    \sum_{j \in N} V_j\bigl(\pi'_j(\varepsilon),\, \pi'_{-j}(\varepsilon)\bigr)
  \;<\;
  \sum_{j \in N} V_j\bigl(\pi^*_j,\, \pi'_{-j}(\varepsilon)\bigr).
\end{align*}
Since $\pi'_j(\varepsilon) = \pi^*_j$ for all $j \neq i$, the terms for $j \neq i$ cancel on both sides, yielding $V_i(\pi_i,\, \pi^*_{-i}) < V_i(\pi^*)$. Since $i$ and $\pi_i \notin \Psi_i(\pi^*)$ were arbitrary, $\pi^*$ is a strict equilibrium.
\end{proof}

\medskip
\begin{proposition}\label{propCharb}
 Lemma \ref{lemmaRepeated} holds in the case $q>0$, which, with the notation of Lemma \ref{lemmaRepeated} is equivalent to the following.
    \begin{enumerate}[label=(\alph*)]
    \item \label{C2partA} Condition \textit{C'(i)} is equivalent to the following condition: for all $i\in N$, for all $e_\alpha$ in the support\footnote{Our result is stated for pure equilibria $\pi^*$, in which case Condition \textit{C'(i)} implies nothing, but we still state what the condition implies in a more general framework.} of $\pi^*$, $V_i(e_\alpha, \pi^*_{-i}) = V_i(\pi^*)$,
    \item \label{C2partB}  The strategy profile $\pi^*$ being a strict equilibrium implies \textit{C'(ii)},
    \item \label{C2partC}  A strategy profile $\pi^*$ satisfying both \textit{C'(i)} and \textit{C'(ii)} is a strict equilibrium.
    \end{enumerate}
\end{proposition}

\begin{proof}
We prove the points in order.

\noindent\textbf{Proof of \ref{propCharb}\ref{C2partA}}
Recall that a strategy profile $\pi'$ satisfies condition \textit{C'(i)} if for any policy profile $\pi\in \Pi^\ell$, we have $\langle v^q(\pi'), \pi-   \pi'\rangle \leq 0$. This is equivalent to having that for all $\pi\in \Pi^\ell$
$$\sum_{i\in N}\sum_{e_\alpha \in E_i} v_{q,i,\alpha}(\pi') \pi_{i,\alpha}\leq \sum_{i\in N}\sum_{e_\alpha \in E_i} v_{q,i,\alpha}(\pi') \pi'_{i,\alpha}$$ 

or, by making explicit the value of $v_{i,\alpha}^q(\pi')$, this is equivalent to:

\begin{multline*}
\sum_{i\in N}\sum_{e_\alpha \in E_i} \left(\pi_{i,\alpha}'^{q} \left ( V_{i}(e_{\alpha}, \pi'_{-i}) - \frac{ \sum_{e_\beta\in E_i}\pi_{i,\beta}'^{q} V_{i}(e_{\beta}, \pi'_{-i})}{\sum_{e_\beta\in E_i}\pi_{i,\beta}'^{q}} \right )\right) \pi_{i,\alpha}\\ 
\leq\sum_{i\in N}\sum_{e_\alpha \in E_i} \left(\pi_{i,\alpha}'^{q} \left ( V_{i}(e_{\alpha}, \pi'_{-i}) - \frac{ \sum_{e_\beta\in E_i}\pi_{i,\beta}'^{q} V_{i}(e_{\beta}, \pi'_{-i})}{\sum_{e_\beta\in E_i}\pi_{i,\beta}'^{q}} \right )\right) \pi'_{i,\alpha}.
\end{multline*}
Note, however, that the expression $\frac{ \sum_{e_\beta\in E_i}\pi_{i,\beta}'^{q} V_{i}(e_{\beta}, \pi'_{-i})}{\sum_{e_\beta\in E_i}\pi_{i,\beta}'^{q}} $ does not depend on the index $\alpha$ of the pure strategy $e_\alpha$, and so we denote this expression by $M_i(\pi')$.

Using this notation, for the policy profile $\pi^*$ we have that \textit{C'(i)} is equivalent to having that for all $\pi\in \Pi^\ell$,

\begin{align}
\notag  \sum_{i\in N}\sum_{e_\alpha \in E_i} &\left(\pi_{i,\alpha}^{*q} \left ( V_{i}(e_{\alpha}, \pi^*_{-i}) - M_i(\pi^*) \right )\right) \pi_{i,\alpha}\\ 
&\leq
\sum_{i\in N}\sum_{e_\alpha \in E_i} \left(\pi_{i,\alpha}^{*q} \left ( V_{i}(e_{\alpha}, \pi^*_{-i}) - M_i(\pi^*)\right )\right) \pi^*_{i,\alpha}. \label{eq_Cii}  
\end{align}

In order to prove \ref{propCharb}\ref{C2partA}, we therefore can show that this new formulation of \textit{C'(i)}, is equivalent to have that for all $i\in N$, for all $e_\alpha$ in the support of $\pi^*$, we have $V_i(e_\alpha, \pi^*_{-i}) = V_i(\pi^*)$.

To do this, we first assume that for all pure strategies $e_\alpha$ in the support of $\pi^*$ we have $V_i(e_\alpha, \pi^*_{-i}) = V_i(\pi^*)$. Because the expression $M_i(\pi^*)$ is a weighted average of such payoffs (which all have the same value), in this case we have $M_i(\pi^*) = V_i(\pi^*)$. Therefore, 
\begin{equation*}
\sum_{i\in N}\sum_{e_\alpha \in E_i} \left(\pi_{i,\alpha}^{*q} \left ( V_{i}(e_{\alpha}, \pi^*_{-i}) - M_i(\pi^*) \right )\right) \pi_{i,\alpha} = 
\sum_{i\in N}\sum_{\substack{
    e_\alpha \in E_i,\\ \pi^*(e_\alpha) > 0}} \left(\pi_{i,\alpha}^{*q} \left ( V_{i}(\pi^*) -V_{i}(\pi^*) \right )\right) \pi_{i,\alpha} = 0
\end{equation*}
and 
\begin{equation*}
\sum_{i\in N}\sum_{e_\alpha \in E_i} \left(\pi_{i,\alpha}^{*q} \left ( V_{i}(e_{\alpha}, \pi^*_{-i}) - M_i(\pi^*) \right )\right) \pi^*_{i,\alpha} = 
\sum_{i\in N}\sum_{\substack{
    e_\alpha \in E_i,\\ \pi^*(e_\alpha) > 0} }\left(\pi_{i,\alpha}^{*q} \left ( V_{i}(\pi^*) -V_{i}(\pi^*) \right )\right) \pi^*_{i,\alpha} = 0
\end{equation*}
Therefore, both sides of Ineq.~\eqref{eq_Cii} are zero, and the inequality holds.

For the other direction, suppose that a strategy profile $\pi^*$ satisfies \textit{C'(i)}, but, for contradiction, 
there is some player $i \in N$ and some pure strategies $e_\alpha\in E_i$ in the support of $\pi_i^*$, such that $V_i(e_\alpha, \pi^*_{-i}) \neq V_i(\pi^*)$. Because $M_i(\pi^*)$ is a weighted average of the payoffs of all pure strategies in the support of $\pi^*_i$, without loss of generality, there exists $e_\beta$ in the support of $\pi^*_i$ such that $V_i(e_\beta, \pi^*_{-i})< M_i(\pi^*) < V_i(e_\alpha, \pi^*_{-i})$.  For $\eta$ small enough, $\pi_i:=\pi_i^* -\eta e_{\beta} + \eta e_{\alpha}$ is a strategy for player~$i$.  The strategy $\pi:=(\pi_i, \pi^*_{-i})$ gives a contradiction to Ineq.~\eqref{eq_Cii}.

We conclude that Condition \textit{C'(i)} is equivalent to the following condition: for all $i\in N$, for all $e_\alpha$ in the support of $\pi^*$, we have $V_i(e_\alpha, \pi^*_{-i}) = V_i(\pi^*)$.

\bigskip

\noindent\textbf{Proof of Part (b)}

\medskip

Suppose a strategy profile $\pi^*$ is a strict equilibrium. We prove that there exists $\epsilon>0$ such that for any 
$\pi\in \Pi^\ell\setminus \Psi(\pi^*)$ at distance at most $\epsilon$ from $\pi^*$ we have $\langle v^q(\pi), \pi-   \pi^*\rangle < 0$. That is, we want to show that if $\pi^*$ is a strict equilibrium then there exists $\epsilon>0$ such that for any $\pi\in \Pi^\ell\setminus \Psi(\pi^*)$ at distance at most $\epsilon$ from $\pi^*$, it holds that:

\[\sum_i\sum_{e_\alpha \in E_i}v_{q,i,\alpha}(\pi)\pi_{i,\alpha}<\sum_i \sum_{e_\alpha \in E_i}v_{q,i,\alpha}(\pi)\pi^*_{i,\alpha}.\]

We prove the above inequality by obtaining, for each player $i$, a lower bound for the right-hand side (RHS), an upper bound for the left-hand side (LHS), and then comparing the bounds. We are therefore going to fix now an $i$ and show that for this $i$ it holds
\[\sum_{e_\alpha \in E_i}v_{q,i,\alpha}(\pi)\pi_{i,\alpha}<\sum_{e_\alpha \in E_i}v_{q,i,\alpha}(\pi)\pi^*_{i,\alpha}.\]

In the following, we expand and re-elaborate equations in non-intuitive ways. To facilitate this operation, we often name recurrent terms of our equations, we therefore are going to define $B_i^{(1)}(\pi), B_i^{(2)}(\pi), \dots{}, T_i^{(1)}(\pi), \dots{}$ and similar notation.

\subsubsection*{Lower bound for RHS:}

Recall that the policy $\pi^*_i$ gives probability $1$ to player $i$ playing $e_{\alpha(i)}^*$ and probability $0$ otherwise. Therefore, for player~$i$,

\begin{align}
\sum_{e_\alpha \in E_i}v_{q,i,\alpha}(\pi)\pi^*_{i,\alpha}&=\sum_{e_\alpha \in E_i} \left(\pi_{i,\alpha}^{q} \left ( V_{i}(e_{\alpha}, \pi_{-i}) - M_i(\pi) \right )\right) \pi^*_{i,\alpha} \notag\\
&=\pi_{i,\alpha(i)}^{q}(V_i(e^*_{\alpha(i)},\pi_{-i})-M_i(\pi)).\label{Eq_2}
\end{align}

We want to lower bound the reward of player $i$ in the case they play $\pi^*_i$, knowing only that the strategy profile $\pi_{-i}$ is close to $\pi^*_{-i}$. We consider different scenarios for $\pi_{-i}$ and use linearity of expectation to get our lower bound. We first consider how likely each scenario is.

We first want to know with what probability, $\pi_{-i}$ is a policy profile in $\Psi_{-i}(\pi^*)$. Let us denote by $Q_j$ the probability that player $j$ plays a pure strategy that is outside of $\Psi_{j}(\pi^*)$, namely $Q_j = \sum_{e_\alpha \in  E_j \setminus \Psi_j(\pi^*)} \pi_j(e_\alpha)$. Therefore with probability $\prod_{j\neq i}(1-Q_j)$ we have that player $i$ receives reward $V_i(\pi^*)$, as $\pi_{-i}$ is contained in $\Psi_{-i}(\pi^*)$. In this scenario, the reward of player $i$ is $V_i(\pi^*)$ by definition of $\Psi_{-i}(\pi^*)$ and therefore the contribution to the expected reward from this scenario is $$V_i(\pi^*)\cdot \prod_{j\neq i}(1-Q_j)= V_i(\pi^*)\left(1-\sum_{j}Q_j+O\big(\sum_{j,k}Q_jQ_k\big)\right).$$

A different case of interest is when exactly one player $j$ plays a strategy outside of $\Psi_j(\pi^*)$, while others conform to $\pi^*$. This happens with probability $Q_j\prod_{k\in N\setminus\{i,j\}}(1-\epsilon_k)$. While we do know this case is not ideal for player $j$, we do not know what happens to the reward of player $i$ in this scenario; we know though with what probability (and therefore weight) this event influences the final result.

To calculate the effect of this factor on the RHS, we introduce the following notation. We denote by $D_{i,j}(\pi)$ the conditional change of reward for player~$i$ conditioning on player~$j$ playing a pure strategy outside $\Psi_j(\pi^*)$ when playing $\pi_j$. Formally, for $Q_j > 0$:

\begin{align*}
   D_{i,j}(\pi) =  V_i(\pi^*) - \frac{\sum_{e_\alpha\in E_j\setminus \Psi_j(\pi^*)}\left[\pi_j(e_\alpha)V_i(e_\alpha, \pi^*_{-j})\right]}{Q_j}.
\end{align*}
Note that, if $Q_i > 0$, $D_{i,i}(\pi)$ is lower bounded by a strictly positive constant. 

Therefore the contribution to the expected reward from this scenario is 
$$\sum_j\left(V_i(\pi^*)-D_{i,j}(\pi)\right)\cdot Q_j\prod_{k\in N\setminus\{i,j\}}(1-\epsilon_k)= V_i(\pi^*)\sum_jQ_j - \sum_{j\neq i} Q_jD_{i,j}(\pi) + O\left(\sum_{j,k}Q_j\varepsilon_k\right).$$

The remaining case happens with the extremely small probability of $1-\prod_{j\neq i}(1-Q_j)-\sum_{j\neq i}Q_j\prod_{k\in N\setminus \{i,j\}}(1-\epsilon_k)$. While the strictness condition of $\pi^*$ doesn't give us any information about $V_i(\pi)$ in this case either, the order of magnitude of the probability is enough for our calculations because the size of the game is bounded and therefore we have a constant bound on the best (and worst) possible rewards for player $i$. As both lower and upper bounds are constant, this factor only accounts for $O(\sum_{j,k}Q_j\epsilon_k)$ in the total sum.

Thus:

$$V_i(\pi^*) - \sum_{j\neq i} Q_jD_{i,j}(\pi) + O\left(\sum_{j,k}Q_j\varepsilon_k\right)= V_i(e^*_{\alpha(i)},\pi_{-i}).$$

The left-hand side of this equation is a value that we are going to use explicitly in the lower bound of $\sum_i \sum_{e_\alpha \in E_i}v_{q,i,\alpha}(\pi)\pi^*_{i,\alpha}$ but also in the upper bound of $\sum_i\sum_{e_\alpha \in E_i}v_{q,i,\alpha}(\pi)\pi_{i,\alpha}$. The comparison of these two bounds is the largest part of this proof. To make the reading easier, we are going to denote by $B^{(1)}_i(\pi)$ the value $V_i(\pi^*) - \sum_{j\neq i} Q_jD_{i,j}(\pi) + O\left(\sum_{j,k}Q_j\varepsilon_k\right)$ (the reason we use this notation is going to be clear when we analyse the upper bound for the LHS). We hope the reader appreciates the readability of the proof over the explicitness of the factors.

Now, rewriting \ref{Eq_2}, factoring in this last inequality, we obtain:

\begin{align*}   
\left(B^{(1)}_i(\pi)- M_i(\pi)\right)(1-\epsilon_i)^q\leq\sum_{e_j} \left(\pi_{i,j}^{q} \left ( V_{i}(e_{j}, \pi_{-i}) - M_i(\pi) \right )\right) \pi^*_{i,j}.
\end{align*}
\medskip

Where we recall that we used the notation $B^{(1)}_i(\pi)=V_i(\pi^*) - \sum_{j\neq i} Q_jD_{i,j}(\pi) + O\left(\sum_{j,k}Q_j\varepsilon_k\right)$.

\medskip

\subsubsection*{Upper bound for LHS:}
We want to bound the value
\begin{align}
  \sum_{e_\alpha \in E_i}v_{q,i,\alpha}(\pi)\pi_{i,\alpha}=\sum_{e_\alpha \in E_i}\pi_{i,\alpha}^{q+1}\left(V_i(e_\alpha, \pi_{-i}) - M_i(\pi)\right).\label{Eq_3}  
\end{align}

Once more we adopt the strategy we followed during the lower bound: we divide the possible values attained by $\pi$ in cases, and we consider with what probability each might happen, knowing that $\pi$ is close in distribution to $\pi^*$. When all players conform to strategy profile $\pi$, the following cases can happen.

\textbf{Player $i$ plays $e^*_{\alpha(i)}$}. This happens with probability $(1-\epsilon_i)$. This case was already considered while doing the lower bound for the RHS, and we know that $V_i(\pi^*) - \sum_{j\neq i}Q_j D_{i,j}(\pi)+O(\sum_{j,k}Q_j\epsilon_k)$ is the value of $V_i(\pi)$ in this case. Let us denote by  $B^{(1)}_i(\pi)=V_i(\pi^*) - \sum_{j\neq i}Q_j D_{i,j}(\pi)+O(\sum_{j,k}Q_j\epsilon_k)$ this first value of $V_i(\pi)$ which holds with probability $(1-\epsilon_i)$.

\textbf{Player $i$ plays a pure strategy in $\Psi_i(\pi^*)\setminus e^*_{\alpha(i)}$}; this happens with probability $(\epsilon_i-Q_i)$. If we condition ourselves to this case, with probability $\prod_{j\neq i}(1-Q_j)$, the opponents are playing an action profile in $\Psi_{-i}(\pi^*)$ and the payoff to player $i$ is $V_i(\pi^*)$. The case where at least one player in $-i$ plays outside of $\Psi_{-i}(\pi^*)$ is bound in probability by $\sum_{j\neq i}Q_j$; remember that the maximum and minimum reward possible for player $i$ are both considered constants as they are fixed beforehand; therefore the expected reward in this sub-case is $O(\sum_{j} Q_j)$. Therefore we can give an upper bound to the reward of player $i$ given that player $i$ plays a pure strategy in $\Psi_i(\pi^*)\setminus e^*_{\alpha(i)}$. This upper bound is:
\begin{align*}
\mathbb{E}\left[V_i(\pi)|\pi_i\in \Psi_i(\pi^*)\setminus e^*_{\alpha(i)}\right] &= \prod_{j\neq i}(1-Q_j)V_i(\pi^*)+O(\sum_jQ_j)\\
&= (1-\sum_{j\neq i} Q_j)V_i(\pi^*) + O(\sum_{j} Q_j ).
\end{align*}
Let us denote by $B^{(2)}_i(\pi)=(1-\sum_{j\neq i} Q_j)V_i(\pi^*) + O(\sum_{j} Q_j )$ this second bound of value of $V_i(\pi)$ which holds with probability $(\epsilon_i-Q_i)$.

\textbf{Player $i$ plays outside of $\Psi_i(\pi^*)$}; this happens with probability $Q_i$. This case can be further split considering that with probability $\prod_{j\neq i}(1-\epsilon_j)$, the opponents play $\pi^*_{-i}$. In this case player $i$ obtains a reward of $V_i(\pi^*)-D_{i,i}(\pi)$.
If the other players do not play according to $\pi^*$, we have no specific bound for this case, but we should always remember that since the game is finite, the general bound for every reward is constant.
Restricting ourselves to the case that player $i$ plays outside of $\Psi_i(\pi^*)$, the expected reward for player $i$ can be upper-bounded considering that with probability $1-\sum_{j\neq i}\varepsilon_j$, the reward for player $i$ has upper bound $(V_i(\pi^*) - D_{i,i}(\pi))$, and considering that the remaining cases can influence the expectation by at most $O(\sum_j\varepsilon_j)$. Therefore we denote by  $$B^{(3)}_i(\pi)=(1-\sum_{j\neq i}\varepsilon_j)(V_i(\pi^*) - D_{i,i}(\pi)) + O(\sum_j\varepsilon_j)$$ this third probabilistic bound of value of $V_i(\pi)$ in the case that player $i$ plays outside of $\Psi_i(\pi^*)$.

Summing up the probabilities of these scenarios and their expected return for player $i$ and substituting them in Eq \ref{Eq_3}, we obtain:
\begin{align*}   
\sum_{e_\alpha \in E_i}\pi_{i,\alpha}^{q+1}\left(V_i(e_\alpha, \pi_{-i}) - M_i(\pi)\right)<\overbrace{ 
(1-\varepsilon_i)^{q+1}\left[B^{(1)}_i(\pi)-M_i(\pi) \right]}^{T_i^{(1)}(\pi)}& \\  
+\overbrace{\left[B^{(2)}_i(\pi)-M_i(\pi)\right]\cdot\hspace{-0.4cm}\sum_{e_\beta\in \Psi_i(\pi^*)\setminus e^*_{\alpha(i)}}\hspace{-0.4cm}\pi^{q+1}_{i,\beta}}^{T_i^{(2)}(\pi)}& \\
 +\overbrace{\left[B^{(3)}_i(\pi) - M_i(\pi)\right] \cdot\sum_{e_\beta\notin \Psi_i(\pi^*)}\pi^{q+1}_{i,\beta} }^{T_i^{(3)}(\pi)}&.
\end{align*}

We now have to analyse and bound $T^{(1)}_i(\pi), T^{(2)}_i(\pi)$ and $T^{(3)}_i(\pi)$. We recall:
\begin{align*}
    B^{(1)}_i(\pi)&=V_i(\pi^*) - \sum_{j\neq i}Q_j D_{i,j}(\pi)+O(\sum_{j,k}Q_j\epsilon_k),\\
    B^{(2)}_i(\pi)&=(1-\sum_{j\neq i} Q_j)V_i(\pi^*) + O(\sum_{j} Q_j ),\\
    B^{(3)}_i(\pi)&=(1-\sum_{j\neq i}\varepsilon_j)(V_i(\pi^*) - D_{i,i}(\pi)) + O(\sum_j\varepsilon_j).
\end{align*}

\textbf{Let us first consider the term $T_i^{(1)}(\pi)$.} For now, we first observe that we can approximate its factor using Taylor expansion, i.e. $(1-\varepsilon_i)^{q+1}= 1- (q+1)\varepsilon_i + O({\varepsilon_i}^2)$.

\textbf{Let us find an upper bound for the term $T_i^{(2)}(\pi)$.} 
The expression $B^{(2)}_i(\pi)-M_i(\pi)$ can be either positive or negative. Whether it is positive or negative and whether $q\geq 1$ or $q<1$ determines whether the upper bound is obtained by concentrating all probability on one pure strategy, thus obtaining weight $(\varepsilon_i-Q_i)^{q+1}$ or evenly distributing it among the relevant pure strategies obtaining the weight $(|\Psi_i(\pi^*) \cap E_i|-1)\left ( \tfrac{\varepsilon_i-Q_i}{|\Psi_i(\pi^*) \cap E_i|-1} \right )^{q+1} = \tfrac{(\varepsilon_i-Q_i)^{q+1}}{(|\Psi_i(\pi^*) \cap E_i|-1)^q}$, but in either case, the term is going to attain its extrema either by putting all the weight in one action, or by distributing it equally as per Jensen's inequality. As these are the only two options, we use the bound
\[K(\varepsilon_i-Q_i)^{q+1}\left[B^{(2)}_i(\pi)-M_i(\pi)\right].\] 
which works for both cases, for some $K \in (0,1]$.

\textbf{Let us find an upper bound for the term $T_i^{(3)}(\pi)$.} We observe that for small enough $\sum_{j\in N}\varepsilon_j$, the third reward, $B^{(3)}_i(\pi)$ is the smallest one amongst $B^{(1)}_i(\pi)$, $B^{(2)}_i(\pi)$ and $B^{(3)}_i(\pi)$, as $D_{i,i}(\pi)$ is lower bounded by a strictly positive constant. Because $B^{(3)}_i(\pi)$ is the smallest of these rewards, and because $M_i(\pi)$ is a weighted average of them, we have that  for small enough $\sum_{j\in N}\varepsilon_j$, we have that $T_i^{(3)}(\pi)$
is a sum of negative values. We can also consider that $\sum_{e_\beta\notin \Psi_i(\pi^*)}\pi_{i,\beta}=Q_i$, and therefore by Jensen's inequality, an upper bound of $T_i^{(3)}(\pi)$ is obtained by letting all summands have the highest reward of pure strategies in $\Psi_i(\pi^*)$, and evenly distribute the probability $Q_i$ among them. This gives the bound:

\begin{align*}
    T_i^{(3)}(\pi)&\leq |E_i\setminus \Psi_i(\pi^*) |\left(\frac{Q_i}{|E_i\setminus \Psi_i(\pi^*)|}\right)^{q+1}\left[B^{(3)}_i(\pi) - M_i(\pi)\right]\\
    &=\frac{(Q_i)^{q+1}}{|E_i\setminus \Psi_i(\pi^*) |^q}\left[B^{(3)}_i(\pi) - M_i(\pi)\right].
\end{align*}

Taking all these cases together, the upper bound for the LHS is, then,

\begin{align*}   
\textstyle \left(1- (q+1)\varepsilon_i + O({\varepsilon_i}^2)\right)\left[B^{(1)}_i(\pi)-M_i(\pi) \right]  &+ K(\varepsilon_i-Q_i)^{q+1}\left[B^{(2)}_i(\pi)-M_i(\pi)\right]\\ 
&+\tfrac{(Q_i)^{q+1}}{|E_i\setminus \Psi_i(\pi^*) |^q}\left[B^{(3)}_i(\pi) - M_i(\pi)\right].
\end{align*}

\subsubsection*{Comparing the bounds:} Putting together the upper bound of the LHS and the lower bound of the RHS that we obtained so far, we have that we need to prove the following inequality.

\begin{align*}
\textstyle \left(1- (q+1)\varepsilon_i + O({\varepsilon_i}^2)\right)\left[B^{(1)}_i(\pi)-M_i(\pi) \right]  &+ K(\varepsilon_i-Q_i)^{q+1}\left[B^{(2)}_i(\pi)-M_i(\pi)\right]\\ 
&+\tfrac{(Q_i)^{q+1}}{|E_i\setminus \Psi_i(\pi^*) |^q}\left[B^{(3)}_i(\pi) - M_i(\pi)\right]\\
 < \left(B^{(1)}_i(\pi)- M_i(\pi)\right)(1-\epsilon_i)^q.
\end{align*}
We now move to the RHS the $\left(B^{(1)}_i(\pi)- M_i(\pi)\right)$ factor in the LHS, to obtain that what we need is equivalent to proving that:

\begin{align*}
 K(\varepsilon_i-Q_i)^{q+1}&\left[B^{(2)}_i(\pi)-M_i(\pi)\right] 
+\tfrac{(Q_i)^{q+1}}{|E_i\setminus \Psi_i(\pi^*) |^q}\left[B^{(3)}_i(\pi) - M_i(\pi)\right]\\
 &< \left(B^{(1)}_i(\pi)- M_i(\pi)\right) \left(\varepsilon_i + O({\varepsilon_i}^2)\right).
\end{align*}

We now divide both sides by $\varepsilon_i$ to obtain that what we need is equivalent to proving that:

\begin{align*}
K\left (1- Q_i/\varepsilon_i \right )&(\varepsilon_i -Q_i)^{q}\left[B^{(2)}_i(\pi)-M_i(\pi)\right]  + \left (Q_i/\varepsilon_i \right ) \tfrac{Q_i^{q}}{|\Pi^\ell\setminus \Psi(\pi^*) |^q}\left[B^{(3)}_i(\pi) - M_i(\pi)\right]\\
& <(1+ O({\varepsilon_i}))\left[B^{(1)}_i(\pi) - M_i(\pi)\right].
\end{align*}

Because the LHS is vanishingly small for small enough $\sum_i\varepsilon_i$, we have that if $0 < B^{(1)}_i(\pi) - M_i(\pi)$, we get our desired inequality. We now expand $M_i(\pi)$, we see that it remains to show:

\begin{equation}\label{eq_Mi}
B^{(1)}_i(\pi) >\frac{ \sum_{\beta}\pi_{i,\beta}^{q} V_{i}(e_{\beta}, \pi_{-i})}{\sum_{\beta}\pi_{i,\beta}^{q}}.
\end{equation}

To prove the inequality, we start by finding an upper bound to $\frac{ \sum_{\beta}\pi_{i,\beta}^{q} V_{i}(e_{\beta}, \pi_{-i})}{\sum_{\beta}\pi_{i,\beta}^{q}}$. We start by splitting the events as before and obtaining, by linearity of expectation applied to $V_i$:

\begin{align}
\frac{ \sum_{\beta}\pi_{i,\beta}^{q} V_{i}(e_{\beta}, \pi_{-i})}{\sum_{\beta}\pi_{i,\beta}^{q}}
\leq \frac{1}{\sum_{\beta}\pi_{i,\beta}^{q}}\left\{(1-\varepsilon_i)^q  \left[B^{(1)}_i(\pi)\right]\right.
&+\sum_{e_\beta\in \Psi_i(\pi^*)\setminus e^*_{\alpha(i)}}\pi^q_{i,\beta}\left[B^{(2)}_i(\pi)\right]\notag
\\
&+\left.\sum_{e_\beta\notin \Psi_i(\pi^*)}\pi^{q}_{i,\beta}\left[B^{(3)}_i(\pi) \right]
\right\}.\label{eq_Mii}
\end{align}

This is a weighted average of the rewards for the different pure strategies, where the sum is split according to the three events of $i$ either playing according to $\pi^*$, or an action in $\Psi_i(\pi^*)$, or outside $\Psi_i(\pi^*)$. To find an adequate upper bound for this weighted average, we consider how the various components of $\pi_i$ influence the whole average, one event at a time.

\textbf{Let us start with the event of $i$ not playing in $\Psi_i(\pi^*)$.}
For $\sum_j\varepsilon_j$ small enough, we showed that $B^{(3)}_i(\pi) $ is the smallest reward. Taking the derivative of the RHS of \eqref{eq_Mii} with respect to $\sum_{e_\beta\notin \Psi_i(\pi^*)}\pi^{q}_{i,\beta}$, we see that to obtain an upper bound for the RHS of \eqref{eq_Mii} we need to minimise $\sum_{e_\beta\notin \Psi_i(\pi^*)}\pi^{q}_{i,\beta}$. For $q\geq 1$, by Jensen's inequality, this is obtained by placing all the probability on one pure strategy, that is $\sum_{e_\beta\notin \Psi_i(\pi^*)}\pi^{q}_{i,\beta}\leq (Q_i)^q$. For $q<1$, also by Jensen's inequality, the upper bound is by dividing the probability equally between all relevant pure strategies:
$\sum_{e_\beta\notin \Psi_i(\pi^*)}\pi^{q}_{i,\beta}\leq \left(\tfrac{Q_i}{|  E_i \setminus \Psi_i(\pi^*)|}\right)^q$. We denote whichever bound is relevant by $K_3^{Q_i}$.

\textbf{For the event of $i$ playing in $\Psi_i(\pi^*)$ but not $\pi^*_i$}, we now take the derivative of the RHS of \eqref{eq_Mii} with respect to $\sum_{e_\beta\in \Psi_i(\pi^*)\setminus e^*_{\alpha(i)}}\pi^q_{i,\beta}$. We have that $B^{(2)}_i(\pi)$ can be either above or below the weighted average of the other rewards. Whether it is above or below that average and whether $q\geq 1$ or $q<1$ determines whether the upper bound is obtained by concentrating all probability on one pure strategy or evenly distributing it among the relevant pure strategies, but in either case, the RHS of \eqref{eq_Mii} is going to attain its extrema either by putting all the weight in one action, or by distributing it equally as per Jensen's inequality. As these are the only two options, we can argue as above and denote the upper bound by $K_2^{Q_i}$, where $K_2^{Q_i}$ can either be $(\varepsilon_i - Q_i)^q$ or $\left(\tfrac{(\varepsilon_i - Q_i)^q}{|  E_i \cap  \Psi_i(\pi^*)|-1}\right)^q$.

If we substitute $K_3^{Q_i}$ and $K_2^{Q_i}$ as factors in Ineq.~\eqref{eq_Mi} and if we rearrange the terms, we obtain that showing Ineq.~\eqref{eq_Mi}, and therefore proving our result, is equivalent to prove the following:
\begin{multline*}
[(1-\varepsilon_i)^q+K_2^{Q_i} + K_3^{Q_i}][B^{(1)}_i(\pi)]\\ 
>(1-\varepsilon_i)^q\left[B^{(1)}_i(\pi)\right]
+K_2^{Q_i}\left[B^{(2)}_i(\pi)\right]
+K_3^{Q_i}\left[B^{(3)}_i(\pi) \right].
\end{multline*}

Cancelling  $(1-\varepsilon_i)^q(B^{(1)}_i(\pi))
$, we get that it suffices to show:
$$
[K_2^{Q_i} + K_3^{Q_i}][B^{(1)}_i(\pi)] >
K_2^{Q_i}\left[B^{(2)}_i(\pi)\right]
+K_3^{Q_i}\left[B^{(3)}_i(\pi) \right].
$$

Taking $\sum_
j \varepsilon_j$ to zero, and making explicit the values of $B_i^{(1)}(\pi), B_i^{(2)}(\pi)$ and $B_i^{(3)}(\pi)$ the inequality becomes:

\[
[K_3^{Q_i} + K_2^{Q_i}]V_i(\pi^*) > 
K_2^{Q_i} V_i(\pi^*)
+K_3^{Q_i}(V_i(\pi^*) - D_{i,i}(\pi)).
\]

Which holds since $D_{i,i}(\pi)>0$. This concludes the proof.

\noindent\textbf{Proof of \ref{propCharb}\ref{C2partC}}
Suppose a strategy profile $\pi^*$ satisfies $C'(i)$ and $C'(ii)$. By contradiction suppose $\pi^*$ is not a strict equilibrium. That is, there is at least one player, player $i$, and at least one pure strategy for player $i$, $\pi'_i\not\in \Psi_i(\pi^*)$, such that $V_i(\pi'_i, \pi^*_{-i})\geq V_i(\pi^*)$. The strategy profile $\pi=(\pi'_i,\pi^*_{-i})$ gives $\langle v^q(\pi),\pi-\pi^*\rangle\geq 0$, a contradiction to $C'(ii)$.

\end{proof}

As a corollary, we have, as a private case of the previous lemma where the length of each player's histories is zero and the set of histories is the empty set: 

\begin{corollary}
\label{lemmaOneShot}
Let $G$ be a one-shot game as previously described. For any $q\geq 0$, a strategy profile $\pi^* \in \prod_{i \in N} \Delta(A_i)$ is a strict Nash equilibrium if and only if the following two conditions are satisfied:
\begin{itemize}
    \item[C(i)] For any strategy profile $\pi\in \prod_{i\in N}\Delta(A_i)$,  we have $\langle v^q(\pi^*), \pi-   \pi^*\rangle \leq 0$.\footnote{A strategy with this property is sometimes known as first-order stationary policy. For $q= 0$, this property is equivalent to the strategy being a Nash equilibrium. Whereas, for $q > 0$, this property is equivalent to each player only randomising over actions between which they are indifferent, which is referred to as a selection equilibrium in \cite{viossat2005correlated}.}
    \item[C(ii)] There is $\epsilon>0$ such that for any strategy profile $\pi\in \prod_{i\in N}\Delta(A_i) \setminus \{ \pi^* \}$ at distance at most $\epsilon$ from $\pi^*$, we have $\langle v^q(\pi), \pi-   \pi^*\rangle < 0$.\footnote{A first-order stationary policy with this second property is also known as stable.}
\end{itemize}
\end{corollary}

\subsection{Proof of Theorem \ref{thm: perfmon_qrepl}} \label{Appendix_proofs_thm_qRep}

We recall the main convergence result for the $q$-replicator dynamics.

\thmPerfMonqRep*

To prove this theorem, we first clarify the assumptions imposed on the $q$-gradient estimator.

Because $\hat{v}^t_i$ is a random variable, in this setting the $q$-replicator dynamics $\pi^t$ is a stochastic process. We write $\mathcal{F}^{t} := \mathcal{F}(\pi^{0},...,\pi^{t})$ for the filtration of the probability space up to and including episode $t$. 
We define 
$$ U^{t} = \hat{v}^{t} - \mathbb{E}[\hat{v}^{t} | \mathcal{F}^{t-1}]  \hspace{1cm}\text{ and }  \hspace{1cm} b^{t} = \mathbb{E}[\hat{v}^{t} | \mathcal{F}^{t-1}] - v^q(\pi^{t}).$$

We assume that $U^{t}$ and $b^t$ are bounded such that: 
$$ \mathbb{E} \left [ \lVert U^t \rVert^{2} | \mathcal{F}^{t-1} \right ] \leq (\sigma^{t})^{2}  \hspace{1cm} \text{and} \hspace{1cm} \mathbb{E} \left [ \lVert b^t \rVert | \mathcal{F}^{t-1} \right ] \leq B^t,$$ 
where $\sigma^t  = \mathcal{O}(t^{\ell_{\sigma}})$ and $B^t = \mathcal{O}(t^{-\ell_{b}})$ for $\ell_{\sigma}, \ell_b > 0$. An example of a procedure that obtains estimators satisfying the above assumptions is REINFORCE that is detailed in Appendix~\ref{sec: reinforce}. Under these assumptions, we obtain the following result, which is a more general result than what is stated in the main text. The main text refers to strict subgame-perfect equilibria, while our result only resorts to strictness. To this end, a definition. 
\begin{definition}[$\ell$-recall strict equilibrium]\label{def_ell_strict}
A strategy profile $\pi^*\in \Pi^{\ell}$ is an $\ell$-recall strict equilibrium if  for any player $i$ and any strategy $\pi_i \in \Pi^{\ell}_i$, we have $V_{i}(\pi^*) > V_{i}(\pi_i,\pi_{-i}^*)$ or $(\pi_i,\pi_{-i}^*) \in \Psi(\pi^*)$. 
\end{definition}

Note that any $\ell$-recall strict equilibrium $\pi^*$ must be deterministic on-path. 

\begin{theorem}[Local convergence of $q$-replicator dynamics]\label{thm:qrep_local}
Let $\pi^*\in \Pi^\ell$ be an $\ell$-recall strict equilibrium, and let $\pi_{t}$ be the sequence of play generated by $q$-replicator learning dynamics with stepsize $\gamma_{t} = \gamma / (t+m)^{p}$, $p \in (1/2,1]$ and $q$-replicator estimates such that $p + \ell_{b} > 1$ and $p - \ell_{\sigma} > 1/2$. Then there exists a neighbourhood $\mathcal{U}$ of $\pi^*$ in $\Pi^\ell$ such that, for any given $\eta>0$ we have: 
$$
    \mathbb{P}\left( \pi_t \rightarrow \Psi(\pi^*) \text{ as } t \rightarrow \infty\right )\geq 1-\eta.
$$
provided that $\gamma$ is small enough (or $m$ is large enough) relative to $\eta$.
\end{theorem}

\begin{proof}[Proof of Theorem~\ref{thm:qrep_local}] The result follows by an adaptation of the proof of Theorem~2 in \cite{giannou2021convergence}. That theorem establishes local almost-sure convergence of projected gradient dynamics to a strict Nash equilibrium of a finite stochastic game under the same stepsize and noise conditions.

Although our setting is that of $\ell$-recall repeated games rather than stochastic games with an exogenous state space, the two models are equivalent once histories are treated as states. Concretely, the set of length-$\ell$ histories forms a finite state space, transitions are deterministic given the current history and action profile, and an $\ell$-recall strategy $\pi\in\Pi^\ell$ is precisely a stationary Markov policy on this induced stochastic game. Under this identification, $\ell$-recall strict equilibria correspond exactly
to strict Nash equilibria in the induced stochastic game.

The stochastic-approximation structure of the policy updates is unchanged by this reformulation, and the stepsize and moment assumptions ensure that all martingale and bias estimates used in \cite{giannou2021convergence} continue to hold. The only substantive difference is that, in the repeated-game setting, some histories may be off-path under $\pi^*$ and hence visited with zero probability. As a result, convergence cannot be guaranteed to a unique policy on the entire history space.

However, inspection of the proof of \cite{giannou2021convergence} shows that the result will hold for any set of policies that satisfy certain variational inequalities for a given dynamic. Once one establishes the appropriate variational inequality characterisation of strict equilibria on the on-path histories for $q$-replicator dynamics, which is precisely the content of Lemma~\ref{lemmaRepeated}, the local Lyapunov and stability arguments apply verbatim and thus apply for $q>0$. Consequently, the learning dynamics converge with arbitrarily high probability to the set $\Psi(\pi^*)$ of policies that coincide with $\pi^*$ on all histories reached under $\pi^*$.

This yields the stated result.
\end{proof}

We now combine Theorem~\ref{thm:qrep_local} and Theorem~\ref{thm: payoff_approx} to establish the main result for $q$-Replicator dynamics.

\begin{proof}[Proof of Theorem~\ref{thm: perfmon_qrepl}]
Fix $\varepsilon > 0$, then by Theorem~\ref{thm: payoff_approx}, there is $\delta^*\in (0,1)$ such that for all $\delta\in (\delta^*, 1)$ and every $u\in \tilde{W}$, there exists $\ell \in \mathbb{N}$ and an $\ell$-recall strict equilibrium $\pi^*$ of $\Gamma(\delta)$ such that the distance between $u$ and  $V(\pi^*)$ is at most $\epsilon$. Take such $\delta\in (\delta^*, 1)$ and, for a fixed $u\in \tilde{W}$, such a $\pi^*$. We now apply Theorem~\ref{thm:qrep_local} to this strict finite-recall equilibrium $\pi^*$. That theorem establishes that, for any given $\eta>0$, there exists a neighbourhood $\mathcal U$ of $\pi^*$ in $\Pi^\ell$ such that, provided the stepsize parameters of the $q$-replicator dynamics are chosen appropriately, the stochastic learning process converges with probability at least $1-\eta$ to the set $\Psi(\pi^*)$ of policies that coincide with $\pi^*$ on all on-path histories.

Since all policies in $\Psi(\pi^*)$ induce the same distribution over play along the equilibrium path, they generate the same expected payoff vector as $\pi^*$. In particular, the payoff vector induced by any limit point of the learning dynamics lies within $\epsilon$ of the target payoff $u$. Therefore, starting from initial conditions in $\mathcal U$, the stochastic $q$-replicator dynamics converge with arbitrarily high probability to an equilibrium outcome whose expected payoffs approximate $u$ to within $\epsilon$.

As  $\epsilon>0$  and $u\in\tilde W$ were arbitrary, this establishes that every feasible and individually rational payoff vector admits an $\ell$-recall strict equilibrium with a basin of attraction under the standard stochastic $q$-replicator dynamics. This completes the proof.
\end{proof}

\section{Approximating the gradient: REINFORCE}\label{sec: reinforce}

In order to implement the $q$-replicator dynamics \eqref{eq: q-replicator}, each player needs to be able to compute $v_{q,i}(\pi^t)$, the $q$-gradient of $V_i$ at $\pi^t$. However, in most practical cases it may not be reasonable to assume that player $i$ knows $\pi_{-i}$, or even their own expected reward function.

This issue may be overcome by player $i$ in episode $t$ having access to an estimator, $\hat{v}^t_i$, of their $q$-gradient $v_{i,q}(\pi^t)$.  We show that any estimator satisfying certain decreasing bounds on bias and variance is sufficient for our convergence result. One such estimator is the well-studied algorithm REINFORCE, which allows each player to compute an unbiased estimation of $v_{i,q}(\pi^t)$ only knowing their strategy $\pi^t_i$ and their realised reward, which we define the following way:

For a realised history $h = (a^1, z^1, a^2, z^2, \dots)$, let
\begin{align*}
    R_i(h) \;=\; \sum_{t = 1}^\infty \,\delta^{t-1}\, R_i(a^t),
\end{align*}
that is, $R_i(h)$ denotes the total discounted reward accumulated by player $i$ along the history $h$, where rewards realised in period $t$ are discounted by the factor $\delta^{t-1}$.

The version of $q$-replicator dynamics that uses the REINFORCE approximation of $v_{i,q}(\pi^t)$ has similar convergence conditions to the one that assumes each player can compute its true value.

In our setting, we have a history $h$ sampled from players playing according to $\pi^t\in\Pi^\ell$. Each player $i$ knows their reward $R_i(h)$ associated with history $h$ and can calculate a value $\Lambda_i(h)$. This value can be interpreted as a measure of the probability of the actions taken by player $i$ that resulted in $h$ being realised. 

% For a detailed explanation of how $\Lambda_i(h)$ is derived, we refer to the Appendix \ref{Appendix_log_trick}.

In this setting, REINFORCE is an algorithm that takes as inputs $R_i(h)$ and $\Lambda_i(h)$ and gives as output an estimate for $v_{i,q}(\pi^t)$. This estimate is unbiased if, for every possible history $h$, we have that $\pi^t$ assigns to every possible action profile a probability bounded away from zero. This is achieved using the $\epsilon$-greedy $q$-replicator dynamics as defined below, which is an example of how REINFORCE can be used in practice by players in the context of $q$-replicator dynamics.

\begin{center}
\begin{minipage}{0.44\textwidth}
\begin{algorithm}[H]
    \centering
    \caption{REINFORCE}\label{algorithm}
    \begin{algorithmic}[1]
        \State \textbf{Input:}    $R_i(h)$, $\Lambda_i(h)$, $\hat{\pi}_i$
        \State $\hat{w}_i\gets R_i(h)\cdot \Lambda_i(h)$ 
        \State $\hat{v}_i\gets \hat{\pi}_{i,j}^q \left ( \hat{w}_{i}(e_{j}) - \frac{ \sum_{k}\hat{\pi}_{i,k}^{q} \hat{w}_{i}(e_{k})}{\sum_{k}\hat{\pi}_{i,k}^q} \right )$ 
        \State
        \textbf{return}  $\hat{v_i}$
    \end{algorithmic}
\end{algorithm}
\end{minipage}
\begin{minipage}{0.55\textwidth}
\begin{algorithm}[H]
    \centering
    \caption{$\varepsilon$-GREEDY $q$-REPLICATOR}\label{algorithm1}
    \begin{algorithmic}[1]
    \State \textbf{Input:} $\pi^0\in \Pi^\ell$, $\{\gamma^t_i\}_{i\in N, t\in\mathbb{N}}$, $\varepsilon\in (0,1)$
        \For{$t=1,2,\ldots$}
        \State $\hat{\pi}^t\gets(1-\varepsilon)\pi^t + \varepsilon$ 
        \State $\text{Sample }h\sim \hat{\pi}^t$
        \For{$i\in N$}
       \State Compute $R_i(h)$, 
       \State $\Lambda_i(h)\gets \sum_{t=0}^{\tau(h)}\nabla_i(\log(\hat{\pi}_i(a_i^t|\hat{h}_i^{\ell_i})))$
        \State $\hat{v}^t_i\gets \text{REINFORCE}(R_i(h), \Lambda_i(h), \hat{\pi}_i^t)$ 
        \State $\pi_i^{t+1}\gets \text{proj}_{\Pi_i}(\pi_i^t+\gamma^t_i\hat{v}_i^t)$
        \EndFor
        \EndFor
    \end{algorithmic}
\end{algorithm}
\end{minipage}
\end{center}

Where $(1-\varepsilon)\pi^t + \varepsilon$ is the strategy profile where for each history $h$, each player $i$ plays $\pi_i^t$ with probability $1-\varepsilon$, and with the remaining probability $i$ plays an action sampled uniformly from $A_i$.

\section{Sequential Equilibria }\label{sec: sequentialequi}

\subsection{A strategy profile with a basin of attraction that is not a sequential equilibrium}

Consider the following variation of prisoners' dilemma, with the actions indexed according to the players: 

\begin{center}
	\begin{tabular}{|l|c|c|}
	    \hline
		\textbf{} & \textbf{$C_2$}& \textbf{$D_2$}\\
		\hline
		\textbf{$C_1$}& $4,4$& $0,5$\\
		\hline
		\textbf{$D_1$}& $5,0$& $2,2$\\
		\hline
	\end{tabular}
\end{center}

Consider a simple symmetric one-recall strategy profile for each player $i \in \{1,2\}$: 
\begin{itemize}
    \item Following the histories $(C_1,C_2)$ or $(D_1,D_2)$ or the empty history, play $C_i$,
    \item Following the history $(D_1,C_2)$ or $(C_1,D_2)$, play $D_i$.
\end{itemize}

With perfect monitoring, this is a subgame-perfect equilibrium for sufficiently patient players. However, we consider a game with imperfect monitoring. To this end, let $c_1$ and $d_1$ ($c_2$ and $d_2$) be the private signals for player~2 (player~1) regarding the actions taken by player~1 (player~2). Suppose that these signals are accurate following the action profiles $(C_1, D_2)$, $(D_1, C_2)$ and $(D_1, D_2)$. That is, following a period when $(C_1, D_2)$ was played, player~1 observes the signal $d_2$ with probability 1 (accurately reflecting the action of player~2), and player~2 observes the signal $c_1$ with probability 1. Similarly, the signal profile following periods when $(D_1, C_2)$ or $(D_1, D_2)$ were played are with probability 1 $(c_2,d_1)$ and $(d_1, d_2)$ respectively.

However, when $(C_1, C_2)$ is played, there is some small probability of inaccurate private signals. More specifically, following a period when $(C_1, C_2)$ was played, the distribution of signals is:

\[
  q((C_1,C_2)) = \left\{
     \begin{array}{@{}l@{\thinspace}l}
       (c_2,c_1)  &: \text{with probability } 1-\varepsilon_1-\varepsilon_2 - \varepsilon_3\\
       (d_2, c_1) &: \text{with probability } \varepsilon_1 \\
       (c_2, d_1) &: \text{with probability } \varepsilon_2\\
       (d_2, d_1) &:\text{with probability } \varepsilon_3 \\

     \end{array}
   \right.
\]

Consider the one-recall strategy for player~1 that begins with playing $C_1$ replies to the action-signal combinations $(C_1, c_2)$ and $(D_1,d_2)$ with $C_1$ and otherwise with $D_1$. For a range of small $\varepsilon$-s, this is the best response to a similar strategy played by player~2 among the one-recall strategies. This is easily computed by considering each one-recall pure strategy of player~1, combined with the strategy of player~2 as an MDP, and finding the stationary distribution of the resulting MDP. Therefore, for this range of $\epsilon$-s, this is a one-recall strict equilibrium and thus has a basin of attraction.

To discuss sequential equilibrium, we should detail the Bayesian updating of beliefs. Suppose player~1 played $C_1$ and observes $d_2$. This can be the result of three situations:
\begin{itemize}
    \item Option 1 - player~2 deviated and played $D_2$ when they should have played $C_2$.
    \item Option 2 - player~2 conformed, played $C_2$, but the signal was wrong.
    \item Option 3 - Player~2, before the previous period, played $C_2$ observed $d_1$ (correctly or incorrectly). Therefore, player~2 is punishing player~1 by playing $D_2$, as they should.
\end{itemize}

If Option~1 takes place, then the best response of player~1 is to play $D_1$. 

If Option~2 occurs, then the best response is to ignore the mistaken signal and play $C_1$. 

If Option~3 took place, player~2 played $D_2$ as they should, and observed $(c_1, D_2)$, then they played $D_2$ and the best response is to play $D_1$.

Suppose player~1, during the first period of the game played $C_1$ and observed $d_2$. Giving an initial probability of 1 to player~2 conforming to the equilibrium, Bayesian updating yields that player~2 surely played $C_2$ and the signal observed is just a monitoring error (Option~2). The best response is to ignore this signal and play $C_1$ in the next period.

However, if player~1 played $D_1$ in the first period, and $C_1$ in the second, and observed $d_2$ in the second, then Bayesian updating gives that Option~3 is the likely one, hence the best response for player~1 is to play $D_1$ in the next period. 

This means that this one-recall strategy profile is not a sequential equilibrium. Indeed, computing the Bayesian probability of Option~3 requires more than one recall. While in sequential equilibrium the Bayesian nature of the updating of beliefs aggregates information as the play unfolds, it cannot be done with one-recall.
\end{document}